\pdfoutput=1
\documentclass[11pt]{mathdoc}
\usepackage[scaled=0.95]{helvet} %
\usepackage[normalem]{ulem}
\hypersetup{linkcolor=black}

\usepackage{mcnotation}

\usepackage{tikz}
\usepackage[shortlabels, inline]{enumitem}
\usepackage{amsfonts,mathtools,dsfont}
\usepackage{xcolor}
\usepackage{IEEEtrantools}
\usepackage{graphicx}
\usepackage{subcaption}
\usepackage{accents}

\usepackage{appendix}
\AddToHook{cmd/appendix/before}{\crefalias{section}{appendix}}
\AddToHook{cmd/appendix/before}{\crefalias{subsection}{appendix}}

\usepackage[noabbrev]{cleveref}
\crefname{assumption}{Assumption}{Assumptions}
\crefformat{equation}{(#2#1#3)}
\crefmultiformat{equation}{(#2#1#3)}{ and~(#2#1#3)}{, (#2#1#3)}{ and~(#2#1#3)}
\crefrangeformat{equation}{equations~(#3#1#4) to~(#5#2#6)}
\crefformat{assumption}{Assumption~#2#1#3}
\crefformat{figure}{Figure~#2#1#3}

\usepackage{acronym}
\acrodef{pro}[{\scshape PrO}]{predictively oriented}
\acrodef{hmc}[HMC]{Hamiltonian Monte Carlo}
\acrodef{mcmc}[MCMC]{Markov chain Monte Carlo}
\acrodef{wgf}[WGF]{Wasserstein gradient flow}
\acrodef{mfld}[MFLD]{mean-field Langevin dynamics}

\newtheorem{theorem}{Theorem}
\newtheorem{proposition}{Proposition}
\newtheorem{lemma}{Lemma}
\newtheorem{corollary}{Corollary}

\newtheorem{assumption}{Assumption}

\usepackage{xcolor}
\definecolor{procolour}{HTML}{4477AA}
\definecolor{gibbscolour}{HTML}{EE6677}
\definecolor{bayescolour}{HTML}{228833}

\title{%
Gibbs representations of predictively oriented posteriors%
}
\author[1]{Yann McLatchie}
\author[1]{Louis Sharrock}
\author[2]{David T. Frazier}
\author[1]{Jeremias Knoblauch}
\affil[1]{Department of Statistical Science, University College London}
\affil[2]{Department of Econometrics and Business Statistics, Monash University}
\makeatletter
\def\runauthor{McLatchie, Sharrock, Frazier, and Knoblauch}
\let\inserttitle\@title
\def\runtitle{\inserttitle}
\makeatother

\begin{document}
\maketitle
\begin{abstract}
    The predictively oriented posterior offers principled  uncertainty quantification, even under model misspecification.
However, it does not admit an explicit density and therefore cannot be computed using classical Monte Carlo sampling algorithms.
We remedy this by deriving an approximation to the predictively oriented posterior whose density can be evaluated point-wise, and whose approximation error decays rapidly.
These results are  illustrated on case studies from epidemiology, spatial statistics, and low-energy nuclear physics.

\end{abstract}
\allowdisplaybreaks
\section{Introduction}
Bayes' theorem optimally updates prior beliefs on unknown parameters to posterior beliefs conditional on observed data \citep{zellner_optimal_1988}. 
When the assumed statistical model agrees with the process that generated the observed data, the Bayesian posterior leads to predictions that are also optimal \citep{aitchison_goodness_1975}.
The accuracy of these predictions, however, is tethered to the assumed model's ability to recover the true data generating process:
when it is misspecified, optimality is no longer guaranteed, and the standard Bayesian posterior ceases to be the natural object of interest.
In response to this, \ac{pro} posteriors recast parameter uncertainty as a function of the induced predictive distribution directly \citep{mclatchie_predictively_2025,shen_prediction-centric_2025,lai_predictive_2024}.
Under some regularity conditions, \ac{pro} posteriors predictively dominate the Bayes posterior and converge to the predictively optimal model.

Translating this theoretically appealing construction into practice has so far been limited by computation.
The \ac{pro} posterior is defined implicitly through a non-linear optimisation programme over the space of probability measures.
This non-linearity, which is central to the very definition of the \ac{pro} posterior, is the main computational obstruction:
the minimiser is characterised \textit{only} by an implicit fixed-point equation.
While variational approaches to the \ac{pro} posterior can sidestep this particular computational hurdle \citep{lai_predictive_2024,masegosa_learning_2020,morningstar_pacm_2022}, they 
introduce biases and approximation errors whose effects on inference cannot be easily understood or controlled \citep{wei_direct_2021}.
This induces a direct trade-off between computational complexity and accuracy, and can violate many of the theoretical benefits demonstrated by \citet{mclatchie_predictively_2025}.
Consequently, the main way to draw samples from \ac{pro} posteriors with theoretical guarantees has been using a \ac{wgf} of the entropy-regularised predictive score \citep{mclatchie_predictively_2025, shen_prediction-centric_2025}. 
The resulting dynamics depend at every time on the current law, and thus cannot be simulated directly.
In practice, this law is replaced by an empirical measure, resulting in an interacting particle system.
Whilst asymptotically exact as both time and particle number tend to infinity, this introduces a number of hyper-parameters for which principled tuning guidance remains limited.

Our central contribution concerns the equilibrium of this particle-discretised \ac{wgf}:
for any fixed number of particles, this stationary distribution induces an explicit Gibbs form on an extended parameter space.
{Indeed, this view makes clear that the \ac{wgf} is nothing more than an overdamped Langevin diffusion on this extended space.
This observation allows \ac{pro} posteriors to be sampled using a much broader family of \acs{mcmc} algorithms.
Samplers can then be chosen to eliminate time-discretisation bias (see, e.g.,~\Cref{fig:exact-finite-marginal}), allow for adaptive tuning and gradient-free methods, and can be more efficient in higher dimensions.} 
The remainder of the paper develops this programme.
In~\Cref{sec:pro} we formalise the definitions of the Bayes, Gibbs, and \ac{pro} posteriors.
We then discuss the natural \ac{wgf} approach to computation for \ac{pro} posteriors in~\Cref{sec:wgf}, and how this leads to a tractable target for Monte Carlo methods in~\Cref{sec:mcmc}.
In~\Cref{sec:theory} we quantify the resulting approximation error, while 
~\Cref{sec:log-score} specialises computation and theory to the case of the logarithmic scoring rule and its tempered approximation.
Our computational approaches are illustrated in~\Cref{sec:experiments} across three settings where misspecification is scientifically consequential, including for spatial data, simulation-based models in physics {where the gradient is unavailable}, and epidemiology.
\section{The predictively oriented posterior}\label{sec:pro}
Let $x_{1:n}\in\mathcal{X}^n$ denote data sampled from the unknown true data-generating process $P_0$.
For a statistical model $\{P_\theta:\,\theta\in\Theta\}$ and a prior $\Pi \in \mathcal{P}(\Theta)$, whose densities with respect to the Lebesgue measure we denote by $p_\theta$ and $\pi$ respectively, the Bayes posterior is
\begin{equation*}
    \pi_n(\theta\mid x_{1:n}) \propto p_\theta(x_{1:n})\pi(\theta).
\end{equation*}
While traditionally motivated through Bayes' rule, the Bayes posterior can also be represented as the minimiser of the variational optimisation problem \citep{knoblauch_optimization_2022,alquier_user-friendly_2024}
\begin{equation}\label{eq:bayes-posterior-defn}
    \Pi_n = \argmin_{Q\in\mathcal{P}(\Theta)}\left\{- \int \log p_\theta(x_{1:n}) \dt Q(\theta)+{\kld{Q}{\Pi}}\right\}.
\end{equation}
The solution to this programme, $\Pi_n$, processes information optimally when the underlying model is well-specified \citep{zellner_optimal_1988,aitchison_goodness_1975}, but is brittle under misspecification \citep{owhadi_brittleness_2015,bissiri_general_2016,jewson_principles_2018, knoblauch_optimization_2022}.

To rectify this shortcoming, a substantial body of work has sought alternative belief updates which are robust to forms of model misspecification \citep{knoblauch_doubly_2018,altamirano_robust_2024,matsubara_robust_2022,altamirano_robust_2023,matsubara_generalized_2023}.
One such class of methods is to replace the logarithmic loss in \cref{eq:bayes-posterior-defn} with alternative loss functions.
An important example of this is scoring rules, functions $S:\mathcal{P}(\mathcal{X}) \times \mathcal{X} \to \mathbb{R}\cup\{\infty\}$ whose empirical counterpart evaluated on data we write as $S_n(P_\theta) = n^{-1}\sum_{i=1}^n S(P_\theta,\,x_i)$. 
Unlike other losses, proper scoring rules always correspond to an associated statistical divergence 
$d_{\mathcal{S}}(P,P^\prime) = \mathcal{S}(P, P^\prime) - \mathcal{S}(P^\prime, P^\prime)$ where $\mathcal{S}(P, P^\prime) = \Ecurl[Y\sim P^\prime]{S(P,Y)}$.

Replacing the log loss in \cref{eq:bayes-posterior-defn} by a scoring rule yields the Gibbs posterior $G_n$ induced by $S$ which is the solution to
\begin{equation}\label{eq:gibbs-posterior-defn}
    G_n = \argmin_{Q\in\mathcal{P}(\Theta)}\left\{\int S_n(P_\theta) \dt Q(\theta)+\frac{\kld{Q}{\Pi}}{\lambda_n}\right\},
\end{equation}
where $\lambda_n$ is referred to as the learning rate and serves to re-weight the influence of data and the prior \citep{mclatchie_predictive_2025,syring_calibrating_2019,holmes_assigning_2017}.
Under the logarithmic score, $S_n(P_\theta) = -n^{-1}\sum_{i=1}^n\log p_\theta(x_i)$, and the choice $\lambda_n = n$, this recovers the Bayes posterior of~\cref{eq:bayes-posterior-defn} as a special case.
Further, due to the Donsker--Varadhan lemma, the solution to~\cref{eq:gibbs-posterior-defn} admits an explicit density much like the Bayes posterior.

For a discrepancy $d$, Bayes and Gibbs posteriors are both geared towards uncertainty quantification for the parameter $\theta^\star = \arg\min_\theta d(P_\theta,\,P_0)$ that best describes the data-generating process, regardless of whether or not the model is well-specified \citep{miller_asymptotic_2021,kleijn_misspecification_2006}.
Under mild conditions on the score and prior, $G_n$ concentrates around $\theta^\star$ so that its corresponding predictive distribution, $P_{G_n}$, becomes indistinguishable from $P_{\theta^\star}$ \citep[see e.g.][Lemma 1]{mclatchie_predictive_2025}.
When the model is misspecified however, $P_{\theta^\star}\ne P_0$, and Gibbs posteriors can be overconfident and poorly calibrated relative to the true data-generating process \citep{grunwald_inconsistency_2017, mclatchie_predictive_2025}.

To remedy this, \citet{mclatchie_predictively_2025} introduced \ac{pro} posteriors.
Unlike Gibbs and Bayes posteriors, these \textit{directly} score the predictive distribution $P_Q = \int P_\theta\dt Q(\theta)$ implied by marginalising a given posterior $Q \in \mathcal{P}(\Theta)$ over the posited model, and are defined as the solution to the variational programme
\begin{equation}\label{eq:pro-posterior-defn}
    \Qn = \argmin_{Q\in\mathcal{P}(\Theta)}\left\{ S_n(P_Q) +\frac{\kld{Q}{\Pi}}{\lambda_n}\right\}.
\end{equation}
Under a mild coercivity condition on the prior (detailed in~\Cref{sec:existence-uniqueness-app}) and for scores which are continuous and bounded from below, this programme admits a unique minimiser.
Unlike Gibbs posteriors, this minimiser need not converge to a point mass when the model is misspecified \citep[Corollary 1]{mclatchie_predictively_2025}.
Unlike Bayes and Gibbs posteriors, however, this minimiser is only characterised implicitly via a fixed-point equation, making it impossible to evaluate the density of $\Qn$ pointwise.
We thus cannot compute the \ac{pro} posterior with standard sampling algorithms.
\section{Wasserstein gradient flows}\label{sec:wgf}
Existing approaches \citep[e.g.,][]{shen_prediction-centric_2025,mclatchie_predictively_2025} instead minimise the objective in an asymptotically exact manner via its \ac{wgf}.
To elucidate this, first rewrite \cref{eq:pro-posterior-defn} as
\begin{equation*}
    \mathcal{F}_n(Q) = \lambda_n S_n(P_Q) - \int \log \pi(\theta) \,\mathrm{d}Q(\theta) + \int \log q(\theta) \mathrm{d}Q(\theta),
\end{equation*} 
where $q$ denotes the density of $Q\in\mathcal{P}(\Theta)$. 
The \ac{wgf} is an absolutely continuous curve of distributions $\{Q_t\}_{t\geq 0}$ which evolves in the direction of steepest descent of the map $Q\mapsto\mathcal{F}_n(Q)$ with respect to the 2-Wasserstein distance.
Its {probabilistic} representation is described by the mean-field Langevin dynamics \citep{hu_mean-field_2021,chizat_mean-field_2022}
\begin{equation}\label{eq:dWFn}
    \mathrm{d}\theta_t  = -\left\{\lambda_n\nabla_{\mathrm{W}_2}\,S_n(P_{Q_t})(\theta_t) - \nabla_\theta\log{\pi}(\theta_t) \right\} \mathrm{d}t + \sqrt 2 \mathrm{d} B_t,  
\end{equation}
where $Q_t=\mathrm{Law}(\theta_t)$, $\nabla_{\mathrm{W}_2}\,S_n(P_{Q_t})(\theta_t)$ denotes the Wasserstein gradient of $S_n(P_{Q_t})$ with respect to $Q_t$ evaluated at $\theta_t$, and $B=\{B_t\}_{t\geq 0}$ is a standard Brownian motion.
The law $Q_t$ is generally intractable, and thus must be approximated. 
The standard approach is to replace $Q_t$ with an empirical measure of $m$ particles, resulting in the following interacting particle system: for $j=1,\dots,m$,
\begin{equation}\label{eq:particle-dWFn}
    \mathrm{d}\theta_{j,t} = -\left[\lambda_n\nabla_{\mathrm{W}_2}\,S_n\{P_{\widehat Q_{m,t}(\theta_{1:m,t})}\}(\theta_{j,t}) - \nabla_\theta\log{\pi}(\theta_{j,t}) \right]\mathrm{d}t + \sqrt 2 \mathrm{d} B_{j,t},
\end{equation}
where $\smash{\widehat{Q}_{m,t}(\theta_{1:m,t}) = m^{-1}\sum_{j=1}^m \delta_{\{\theta_{j,t}\}}}$,  with $\delta_{\{x\}}$ the Dirac delta at $x$.\footnote{For an implementation, see the \texttt{pymc-prop} package at \url{https://github.com/pymc-devs/pymc-prop}.}
Under suitable conditions, this empirical measure does indeed approximate the original measure of interest for sufficiently many particles: $\widehat{Q}_{m,t}(\theta_{1:m,t})\rightarrow Q_t$ as $m\to\infty$ \citep[e.g.,][]{suzuki_convergence_2023}.

While asymptotically exact as both $m$ and $t$ diverge, this sampling algorithm introduces several hyper-parameters that require careful tuning in practice.
The most important of these are the step size for the time-discretisation of the underlying stochastic differential equation \citep[see e.g.][Theorem 2, {and~\Cref{sec:step-size} for a numerical demonstration}]{suzuki_convergence_2023} and the number of particles $m$.
The literature offers little principled guidance on how these should be chosen in our setting.
Notably, adaptive schedules for related dynamics such as those in \citet{sharrock_tuning-free_2025} do not transfer immediately and, with the notable exception of \citet{chazal_computable_2025}, there are also very few diagnostic tools by which to judge the quality of any given choice.
\section{Markov chain Monte Carlo}\label{sec:mcmc}
The interacting-particle system of~\cref{eq:particle-dWFn} admits another interpretation which is particularly useful for computation.
Namely, for any fixed number of particles $m$, it is an over-damped Langevin diffusion on the product space $\Theta^m$.
It thus admits an explicit Gibbs stationary distribution whose density can be evaluated pointwise. 
This unlocks sampling via {any} standard {\acs{mcmc}} method, and entirely bypasses the computational difficulties of \ac{wgf}s.
While this finite-particle Gibbs representation has previously been recognised in related settings \citep{chen2025uniform,kook2024sampling,shen_prediction-centric_2025}, we are the first to analyse it systematically in the \ac{pro} posterior setting.

To make this explicit, let $(\theta_1,\ldots,\theta_m)=\theta_{1:m}\in\Theta^m$ be a collection of particles defining an empirical measure $\widehat Q_m(\theta_{1:m})$ whose induced predictive distribution is $P_{\widehat{Q}_{m}(\theta_{1:m})} = m^{-1}\sum_{j=1}^m P_{\theta_j}$.
Each particle configuration $\theta_{1:m}$ then induces the finite-dimensional loss
\begin{equation}\label{eq:particle-loss}
    \Lnm(\theta_{1:m}) = S_n\{P_{\widehat{Q}_{m}(\theta_{1:m})}\} = \frac{1}{n} \sum_{i=1}^n S\left(\frac{1}{m}\sum_{j=1}^m P_{\theta_j},\,x_i\right).
\end{equation}
Whenever the gradient exists, the chain rule implies
$
    m\nabla_{\theta_j} \Lnm(\theta_{1:m}) = \nabla_{\mathrm{W}_2}\,S_n\{P_{\widehat Q_m(\theta_{1:m})}\}(\theta_j).
$
It follows that \cref{eq:particle-dWFn} is exactly the over-damped Langevin diffusion for the potential
\begin{equation*}
    \theta_{1:m} \mapsto m\lambda_n \Lnm(\theta_{1:m}) - \sum_{j=1}^m \log \pi(\theta_j).
\end{equation*}
Under some mild regularity conditions, the interacting particle system of~\cref{eq:particle-dWFn} thus admits the stationary Gibbs measure \citep[Proposition 4.6]{pavliotis_stochastic_2014},
\begin{equation}\label{eq:particle-gibbs-density}
    q_{n,m}(\theta_{1:m}\mid x_{1:n}) \propto \exp\{-m\lambda_n \Lnm(\theta_{1:m})\} \pi^{\otimes m}(\theta_{1:m}).
\end{equation}
Since this is nothing more than a Gibbs measure on the extended space $\Theta^m$, the conditions required for its existence are weaker than those required for the \ac{pro} posterior of Section~\ref{sec:pro} (see \Cref{sec:existence-uniqueness-app}).
Moreover, it can also be expressed as the unique optimum of the linear variational objective \citep[Lemma 2.2]{alquier_user-friendly_2024}
\begin{equation}\label{eq:particle-objective}
    Q_{n,m}=\argmin_{Q_m\in\mathcal{P}(\Theta^m)} \int \Lnm(\theta_{1:m}) \dt Q_m(\theta_{1:m}) + \frac{\kld{Q_m}{\Pi^{\otimes m}}}{m\lambda_n}.
\end{equation}
In this sense, the particle approximation converts the original measure-valued objective in~\cref{eq:pro-posterior-defn} into an ordinary linear variational objective on the extended space $\Theta^m$.

To approximate the \ac{pro} posterior $Q_n$, it remains to relate $Q_{n,m}\in\mathcal{P}(\Theta^m)$ to a suitable distribution in $\mathcal{P}(\Theta)$.
To this end, we define the \emph{particle \ac{pro} posterior} as
\begin{equation*}
    \barQnm = \int \frac{1}{m}\sum_{j=1}^m \delta_{\{\theta_j\}} \dt \Qnm(\theta_{1:m}) = \frac{1}{m}\sum_{j=1}^m \Qnm^{(j)}
\end{equation*}
where $\Qnm^{(j)}(\theta) = \int \Qnm(\theta_{1:m})\dt\theta_{-j}$ denotes the common $j$-th one-particle marginal of $\Qnm$.
Given $B$ samples $\{\theta_{1:m}^{(b)}\}_{b=1}^B$ from $\Qnm$, we can compute expectations with respect to $\barQnm$ by averaging over both the $B$ draws and the $m$ particle coordinates.
For the predictive, this yields the approximation $P_{\barQnm}\approx (mB)^{-1}\sum_{j=1}^m\sum_{b=1}^B P_{\theta_j^{(b)}}$.

This perspective opens the door to {a broad family of algorithms for computation}:
target the density in \cref{eq:particle-gibbs-density} directly with {any \acs{mcmc} algorithm} and average over particles. 
Doing so has numerous advantages.
First, {some \acs{mcmc} algorithms can set} the step size automatically during a warm-up phase to meet a target acceptance rate \citep{roberts_weak_1997, hoffman_no-u-turn_2014}. 
Second, there exists a mature suite of convergence diagnostics for Monte Carlo samplers \citep{gelman_inference_1992, geyer_practical_1992}.
Further, a Metropolis--Hastings correction renders $\Qnm$ the exact invariant distribution, so that the time-discretisation bias of the flow is removed.
Finally, the approach is more lenient in its requirements on $\Lnm$:
in situations where the Gibbs density of~\cref{eq:particle-gibbs-density} is non-differentiable, gradient-free schemes such as random-walk Metropolis--Hastings \citep{metropolis_equation_1953, hastings_monte_1970} can be used.
It is not our ambition to prescribe any single algorithm, and indeed the relative merits of unadjusted and Metropolis-adjusted samplers are well-understood \citep[see, e.g., Section 3.3 of][and references therein]{durmus_asymptotic_2023}.
In particular, unadjusted dynamics such as the \ac{wgf} avoid the step-size shrinkage that Metropolis-adjusted schemes induce in high dimensions, and can accommodate stochastic gradients for large datasets \citep{bou-rabee_mixing_2023, bou-rabee_convergence_2023, suzuki_convergence_2023}.
Metropolis-adjusted samplers on the other hand are preferable when high accuracy is necessary, which comes at only a logarithmic cost once warm-started \citep{dwivedi_log-concave_2019, wu_minimax_2022}.
More recent work has also looked to quantify the conditions under which one Metropolis-adjusted sampler will converge faster than another \citep{chen_metropolized_2026}.

Naturally, the convenience of {sampling from~\cref{eq:particle-gibbs-density}} comes at the cost of {the approximation by finitely-many particles}.
We now quantify its magnitude and rate of decay.
\section{Theoretical guarantees}\label{sec:theory}
\subsection{Finite-particle bias}\label{sec:approx-bias}
In this section we show that the particle \ac{pro} posterior $\barQnm$ is close to the exact \ac{pro} posterior $Q_n$ under a mild convexity assumption that is satisfied for most scores of interest.
\begin{assumption}\label{ass:convex}
    The scoring rule $S(\cdot,x)$ is convex in its first argument for all $x\in\mathcal{X}$, and its corresponding divergence $d_{\mathcal{S}}(\cdot,\,\cdot)$ is jointly convex in both arguments.
\end{assumption}
Let $Q^{\otimes m}=\bigotimes_{j=1}^{m}Q$ denote the $m$-fold product measure of $Q$, and define the \textit{finite-particle bias} as
\begin{equation}\label{eq:beta-defn}
    \beta_{n,m}(Q)
    =
    \int  \Lnm(\theta_{1:m}) \dt Q^{\otimes m}(\theta_{1:m})
    -
    S_n(P_Q).
\end{equation}
This term quantifies the difference between the finite-particle objective of~\cref{eq:particle-objective} and the exact \ac{pro} objective in~\cref{eq:pro-posterior-defn}. 
Our first result is to show that this bias is always non-negative. 
This is essentially immediate from Jensen's inequality; see~\Cref{sec:non-negativity} for a formal proof.
{
\begin{lemma}\label{lem:beta-nonnegative}
    Under \Cref{ass:convex}, $\beta_{n,m}(Q)\ge 0$ for all $Q\in\mathcal{P}(\Theta)$.
\end{lemma}
}
It follows that the finite-particle loss in \eqref{eq:particle-objective} is an upper bound on the exact \ac{pro} objective in~\cref{eq:pro-posterior-defn}. 
We now assume that this bias is vanishing at a controllable rate in the number of particles. 
It is worth emphasising that this is not a restrictive condition. 
In particular, it is satisfied in the two most common cases of interest:
the logarithmic score, under a technical moment condition, and kernel scoring rules;
see~\Cref{lemma:log-score-beta,lemma:di-beta} in the Appendix. 
\begin{assumption}\label{ass:non-trivial-misspecification-exact}
    The finite-particle bias for all $Q\in\mathcal{P}(\Theta)$ is vanishing in the number of particles: 
    $\beta_{n,m}(Q) \lesssim M_n(Q)/m$,
    for some $M_n(Q)>0$.
\end{assumption}
We next establish convergence of the marginal posterior ${\barQnm}$ towards the underlying \ac{pro} posterior ${\Qn}$, at a rate controlled by the finite-particle bias. 
\begin{lemma}\label{lemma:kl-closeness}
    Suppose that~\Cref{ass:convex} holds and that \cref{eq:pro-posterior-defn} admits a unique minimiser.
    Then, 
    $d_{\mathrm{KL}}(\barQnm;\,\Qn)\leq\lambda_n\beta_{n,m}(\Qn)$.
    Additionally, if~\Cref{ass:non-trivial-misspecification-exact} holds at $Q_n$, then 
    $d_{\mathrm{KL}}(\barQnm;\,\Qn)\lesssim\lambda_n / m$.
\end{lemma}
\Cref{lemma:kl-closeness} says that the particle \ac{pro} posterior is close to the ideal \ac{pro} posterior in Kullback--Leibler divergence.
Moreover, under~\Cref{ass:non-trivial-misspecification-exact}, this error vanishes rapidly in $m$.
This result can be further extended to show convergence between $k$-marginals of $\Qnm$ and the $k$-product law of $\Qn$; 
see~\Cref{sec:kl-closeness-supp}.
These results are significant in their own right: they mirror existing propagation of chaos type results \citep[e.g.,][]{kook2024sampling}, but are here obtained under weaker conditions.
We refer to~\Cref{sec:poc-supp} for a more detailed discussion on this point. 
{A similar result was shown by \citet[Theorem 1]{nitanda_improved_2024}, albeit under slightly different assumptions.}
\subsection{Asymptotic properties}
The results above provide asymptotic guarantees for the particle \ac{pro} posterior in the limit as $m\rightarrow\infty$.
We now analyse its convergence in $n$.
In the following, for a given sequence of random variables $\{X_n\}_{n\in\mathbb{N}}$ and scalars $\{a_n\}_{n\in\mathbb{N}}$, we write $X_n\lesssim a_n$ to mean that $P_0(\lim_{n\to\infty}X_n/a_n \leq C) = 1$ for some constant $C > 0$.
Further, if both $X_n\lesssim a_n$ and $a_n\lesssim X_n$ then we write $X_n \asymp a_n$.
We denote expectations taken with respect to $P_0$ by $\mathbb{E}_0$.

To prove concentration of our approximation's predictive distribution to the population optimal predictive, we leverage a sufficient condition which was previously used in \citet{mclatchie_predictively_2025} to derive foundational theoretical guarantees for \ac{pro} posteriors.
\begin{assumption}\label{ass:entropy}
    Let $r_n\downarrow0$ be a decreasing positive sequence, then for $\mathcal{P}_\Pi = \{Q\in\mathcal{P}(\Theta):Q\ll\Pi\}$,
    \begin{equation*}
        \mathbb{E}_0 \sup_{Q\in\mathcal{P}_\Pi} \vert S_n(P_Q) - \Ecurl[0]{S_n(P_Q)} \vert \lesssim r_n.
    \end{equation*}
\end{assumption}
That is, we require the empirical score to concentrate around its expectation at rate $r_n$ uniformly on $\mathcal{P}_\Pi$.
This assumption is known to hold with $r_n\asymp \log(n) / n^{1/2}$ in the case of kernel scoring rules with bounded kernels, and for the log score under some further restrictions \citep{mclatchie_predictively_2025}. 
Additionally, we assume the existence of a predictively optimal distribution $Q^\star$ that satisfies mild regularity conditions.
\begin{assumption}\label{ass:concentration_existence}
    There exists $Q^{\star}\in\arg\inf_{Q\in\mathcal{P}(\Theta)} d_{\mathcal{S}}(P_Q,\,P_0),$ such that $\kld{Q^\star}{\Pi}<\infty$, and $\mathbb{E}_0\{ M_n(Q^\star)\}<\infty$.
\end{assumption}
That is, there exists a measure $Q^\star$ which is population-optimal in terms of its induced predictive distribution $P_{Q^\star}$, which is absolutely continuous with respect to the prior $\Pi$, and for which the finite-particle bias of~\cref{eq:beta-defn} is finite in expectation.
With this, our first result shows that, even though $\barQnm$ is an approximation to the \ac{pro} posterior $Q_n$, it produces eventually indistinguishable predictive distributions as the number of particles is chosen to be sufficiently large.
\begin{theorem}\label{thm:predictive-concentration}
    Under~\Cref{ass:convex,ass:non-trivial-misspecification-exact,ass:entropy,ass:concentration_existence}, for $n$ sufficiently large but finite, taking $\lambda_n\asymp n^{1/2}/\log(n)$,
    \begin{equation*}
        \Ecurl[0]{d_{\mathcal{S}}(P_{\barQnm}, P_0)} \lesssim 
        d_{\mathcal{S}}(P_{Q^\star}, P_0) 
        + \max\left\{\frac{\log(n)}{n^{1/2}},\,r_n\right\}+ \frac{1}{m}.
    \end{equation*}
\end{theorem}
This result compares the divergence between the predictive distribution induced by the particle \ac{pro} posterior and the true data-generating process.
The bound decomposes into three interpretable contributions:
the first term is the irreducible predictive error of the model class; 
the second represents the finite sample error and corresponds to the parametric rate $\log(n)/n^{1/2}$ under relatively mild regularity conditions;
and the last is the bias due to the approximation of $\Qn$ with $\barQnm$.
Critically, the first two terms are unaffected by the approximation, and also appear in the 
bounds derived by \citet{mclatchie_predictively_2025}:
relative to the generalisation guarantees of exact \ac{pro} posteriors, our approximation pays only an incremental penalty vanishing rapidly in $m$.
This is particularly notable since the user directly controls the choice of $m$.
When concentration of the particle \ac{pro} posterior itself is of interest, it is possible to translate these results on the predictive distribution back to the posterior;
see~\Cref{sec:posterior-concentration}.
\section{The logarithmic score}\label{sec:log-score}
\subsection{Connection to mixture modelling}
In the case of the logarithmic score, the particle \ac{pro} posterior based on $m$ particles bears some similarity to a Bayesian mixture model.
Namely, taking $\alpha = {m\lambda_n}/{n}$ the particle \ac{pro} posterior of~\cref{eq:particle-gibbs-density} becomes 
\begin{equation}\label{eq:replica-log-target}
    q_{n,m}^{(\alpha)}
    (\theta_{1:m}\mid x_{1:n}) \propto \prod_{i=1}^{n}\left\{\frac{1}{m}\sum_{j=1}^{m}p_{\theta_j}(x_i)\right\}^\alpha
    \pi^{\otimes m}
    (\theta_{1:m}).
\end{equation}
That is, it is the $\alpha$-tempered Bayes posterior based on an equally-weighted mixture likelihood with density $m^{-1}\sum_{j=1}^{m}p_{\theta_j}$ whose parameters are endowed with the fully-factorised prior $\Pi^{\otimes m}(\theta_{1:m})$.
This relates to the intuition that the \ac{pro} posterior of~\cref{eq:pro-posterior-defn} is an estimate of a nonparametric mixture over model parameters.
In turn, our approximation of it with $m$ particles can be heuristically understood as an atomic approximation of that nonparametric mixture.
This similarity with standard Bayesian techniques allows us to leverage more classical theoretical and computational tools for \ac{pro} posteriors, which we present below.

There remain, however, some fundamental conceptual differences to mixture modelling which are worth clarifying.
Unlike in standard mixture modelling, here we do not need to model the mixture, and there is no notion of latent allocation in our setting.
Similarly, in our case there are no concerns with identifiability.
In this sense, despite some superficial similarities to mixture modelling, the objective of~\cref{eq:replica-log-target} does not share the interpretations standard to that literature.
\subsection{Parallel particle updating}
\label{sec:replica-gibbs}
Building on~\cref{eq:replica-log-target}, we now identify a latent-variable representation that enables parallel updates in the case where $\alpha=m\lambda_n/n$ is a positive integer.
For each $x_i$ we sample $\alpha$ artificial replicas $\tilde x_{1:\alpha}$ and corresponding allocation variables $z_{i,s}$ for $i=1,\ldots,n,\,s=1,\ldots,\alpha$ according to the hierarchical model
\begin{equation}\label{eq:replica-augmentation}
    \theta_j\sim\Pi, \quad
    z_{i,s}\sim\mathrm{Uniform}\{1,\ldots,m\}, \quad
    \tilde x_{i,s}\mid z_{i,s},\theta_{1:m}\sim P_{\theta_{z_{i,s}}}.
\end{equation}
Although this construction introduces $n\alpha$ allocation variables, they enter the parameter conditionals only through the counts
$
    c_{i,j}=\sum_{s=1}^\alpha\mathbf{1}_{\{z_{i,s}=j\}}.
$
Write the count vectors $c_i=(c_{i,1},\ldots,c_{i,m})$, and let 
$
    \widetilde Q_{n,m,\alpha}(\theta_{1:m},c_{1:n}\mid x_{1:n}) 
$
denote the joint posterior distribution over $\theta_{1:m}$ and the count vectors $c_{1:n}$, conditioned on $\alpha$ replicates of the observations $x_{1:n}$.
The following result gives its full conditional distributions.
\begin{proposition}\label{prop:replica-gibbs}
    Under the augmented model in~\cref{eq:replica-augmentation}, the marginal distribution of $\theta_{1:m}$ under $\widetilde Q_{n,m,\alpha}$ is exactly $\Qnm$ in~\cref{eq:replica-log-target}.
    Moreover, conditional on $\theta_{1:m}$, the count vectors $c_i$ are independent and 
    \begin{equation}\label{eq:replica-count-conditional}
        c_i\mid\theta_{1:m},x_{1:n}
        \sim
        \operatorname{Multinomial} \left(\alpha;\rho_{i,1},\ldots,\rho_{i,m}\right),
        \quad
        \rho_{i,j}
        =
        \frac{p_{\theta_j}(x_i)}{\sum_{\ell=1}^m p_{\theta_\ell}(x_i)}.
    \end{equation}
    Conditional on the counts, the particles are independent with
    \begin{equation}\label{eq:replica-component-conditional}
         \tilde q_{n,m,\alpha}
        (\theta_j\mid c_{1:n},x_{1:n})
        \propto
        \left\{\prod_{i=1}^n p_{\theta_j}(x_i)^{c_{i,j}}\right\}
         \pi 
        (\theta_j).
    \end{equation}
\end{proposition}
The resulting blocked 
Gibbs sampler alternates between independent draws of the $n$ count vectors and independent, parallel updates of the $m$ particles.
The $n\alpha$ individual allocations therefore never need to be stored or updated.
For an exponential family model $P_\theta$ equipped with a conjugate prior, the particle conditionals in~\cref{eq:replica-component-conditional} are also closed-form densities (detailed in~\Cref{sec:exponential-family-updates}).
{And thus, at least in this conjugate case, mixing quality is decoupled from the number of particles $m$.}
Meanwhile, in the absence of conjugacy, one can instead sample from a Markov kernel invariant for the conditionals, leading to a Metropolis-within-Gibbs scheme.

\begin{figure}[t!]
    \centering
    \includegraphics[width=\linewidth]{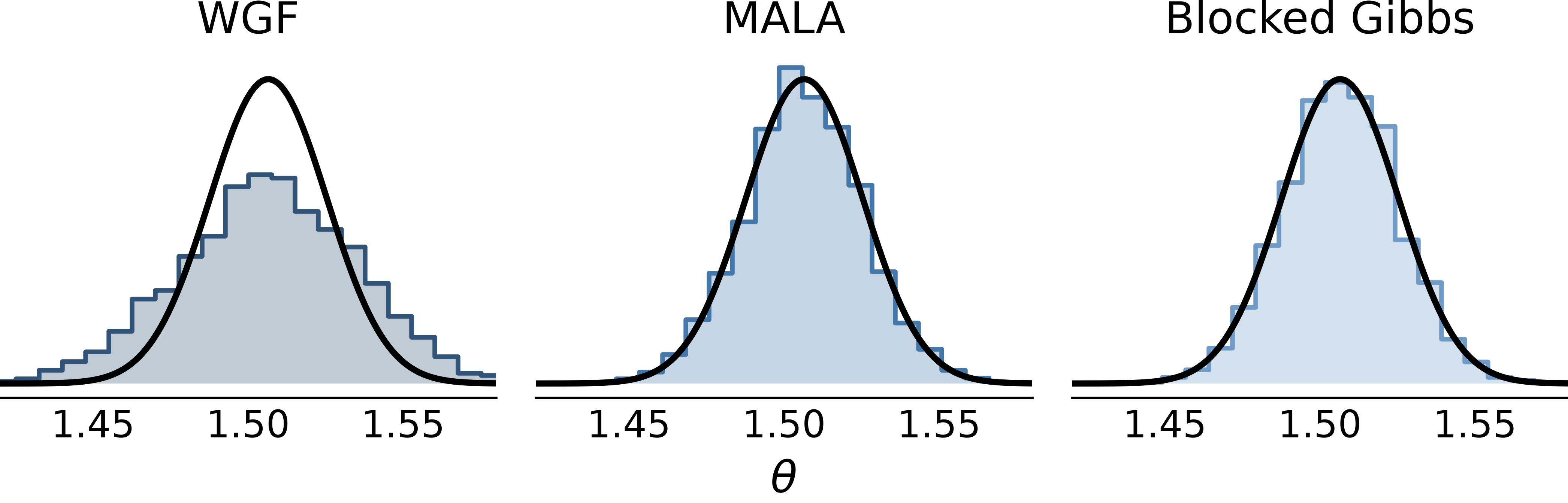}
    \caption{The exact marginal posterior density (solid black curve) alongside a histogram of posterior samples under three different samplers.
    Only the upper mode is shown for clarity.}
    \label{fig:exact-finite-marginal}
\end{figure}
To illustrate this computational approach we fit the Gaussian $P_{\theta}=\normal(\theta,\sigma^2)$ with $\sigma^2=0.2^2$ fixed and known to $n=100$ observations sampled from $P_0 = 0.5 \normal(-1.5,\sigma^2) + 0.5 \normal(1.5,\sigma^2)$. 
For the prior $\Pi = \normal(0,\tau^2)$ with $\tau^2=5^2$, and defining $N_j=\sum_{i=1}^n c_{i,j}$ and $T_j=\sum_{i=1}^n c_{i,j}x_i$, the parameter update for the blocked Gibbs sampler becomes
\begin{equation*}
    \theta_j\mid c_{1:n},x_{1:n}
    \sim
    \normal(\mu_j,v_j^2),
    \quad
    v_j^2 = \left(\tau^{-2}+N_j\sigma^{-2}\right)^{-1},
    \quad
    \mu_j = v_j^2\sigma^{-2}T_j
    \quad(j=1,\ldots,m).
\end{equation*}
A complete sweep then consists only of $n$ multinomial draws and $m$ Gaussian draws.
We use only $m=2$ particles to allow for numerical computation of the marginal density $\dt\barQnm(\theta)$.
In~\Cref{fig:exact-finite-marginal} we compare this marginal posterior density with draws from the blocked Gibbs sampler, from time-discretisation of the \ac{wgf}, and from the Metropolis adjusted Langevin algorithm (MALA) all targetting~\cref{eq:particle-gibbs-density}.
Here, as expected, MALA removes the bias due to time-discretisation of the \ac{wgf}.
The posterior draws from the blocked Gibbs algorithm also align closely with $\dt\barQnm(\theta)$.
We investigate the relative computational efficiency of the blocked Gibbs sampler over MCMC in~\Cref{sec:guerry}.
\subsection{Tempered approximations}\label{sec:temp}
Since $m^{-1}\sum_{j=1}^{m}p_{\theta_j}$ is a density function, 
the Gibbs measure $Q_{n,m}^{(\alpha)}$ in~\cref{eq:replica-log-target} defines a tempered posterior based on the equally-weighted mixture likelihood over $\Theta^m$. 
Taking $\alpha\in(0,1)$ implies that the parallel updating scheme developed above cannot be applied. 
However, for $\alpha\in(0,1)$ the resulting posterior can be interpreted as a tempered posterior. 
By leveraging the structure of tempered posteriors, we can establish generalisation guarantees of the approximation $\barQnm^{(\alpha)}$ when $\alpha\in(0,1)$.
{While these guarantees will be similar to~\Cref{thm:predictive-concentration}, which hold for \textit{any} value of $\alpha$, taking $\alpha\in(0,1)$ allows us to substantially weaken the required regularity conditions.}

To facilitate this result, we impose a mild technical condition from the literature on variational inference \citep[see][]{cherief-abdellatif_asymptotics_2025} which ensures a finite fourth moment for the normalised weights $p_\theta(X) / \int p_\theta(X)\dt Q^\star(\theta)$ for $X \sim P_0$.
This assumption, given formally as~\Cref{ass:log_score} in the Appendix, is sufficient to control the finite-particle bias in line with~\Cref{ass:non-trivial-misspecification-exact} (see~\Cref{lemma:log-score-beta} in the Appendix).
The derived generalisation bound then relies on R\'enyi's $\alpha$-divergence \citep{renyi_measures_1960}, which 
for two measures $P,Q\in\mathcal{P}(\mathcal{X})$ which are absolutely continuous with respect to a third measure $\mu\in\mathcal{P}(\mathcal{X})$ is defined as 
\begin{equation}\label{eq:renyi}
    d_\alpha(P;\,Q) = \frac{1}{(\alpha-1)}\log \int \left(\frac{\dt P}{\dt\mu}\right)^\alpha\left(\frac{\dt Q}{\dt\mu}\right)^{1-\alpha}\dt \mu.
\end{equation}
Notably, this discrepancy converges to the Kullback--Leibler divergence pointwise as $\alpha\rightarrow1$. 
Equipped with this, we finally obtain the following bound in the case of the logarithmic score.
\begin{theorem}\label{thm:predictive-concentration-kl}
    Under {a moment assumption on the normalised weights above}, for a fixed $\alpha\in(0,1)$ and some $C>0$,
    \begin{equation*}
        \mathbb{E}_0\left[ d_\alpha\left\{{P_{\barQnm^{(\alpha)}}};\,{P_0}\right\} \right]\leq \frac{\alpha}{1-\alpha}\left\{ \kld{P_0}{P_{Q^\star}} +\frac{C}{m}\right\}+\frac{m\,\kld{Q^\star}{\Pi}}{n(1-\alpha)}.
    \end{equation*}
 \end{theorem}
This bound has the same constituent parts as in~\Cref{thm:predictive-concentration}, and similarly suggests a rate at which the number of particles should grow to deliver an asymptotically tight bound.
Interestingly, choosing $m\asymp n^{1/2}$ shows that the predictive distribution implied by our approximation concentrates at the usual parametric rate of $n^{-1/2}$ under significantly weaker assumptions than the  results derived for the \textit{exact} \ac{pro} posterior of \cite{mclatchie_predictively_2025}. 
In particular, the latter  required control on an  empirical process over a subset of $\mathcal{P}(\Theta)$, while the above only used a moment assumption.
{\Cref{thm:predictive-concentration-kl} implies that asymptotic sharpness of the particle \ac{pro} posterior requires $m$ to scale like $m\asymp n^{1/2}$.
This may seem computationally troublesome since the mixing properties of Metropolis-adjusted \acs{mcmc} samplers can degrade as $m$ diverges.
This requirement is not, however, problematic
in small-to-moderate sample size settings:
$n$ on the order of tens of thousand only requires on the order of one hundred particles. 
For example, in~\Cref{sec:guerry} we consider values of $m$ up to $256$.
This is a far larger number than~\Cref{thm:predictive-concentration-kl} would suggest, and is theoretically sufficient for sample sizes up to $n\le 65{,}500$.}

\section{Applications}\label{sec:experiments}
\subsection{Compartmental modelling of simulated epidemiological data}\label{sec:sir}
\begin{figure}[t!]
    \centering
    \includegraphics[width=\linewidth]{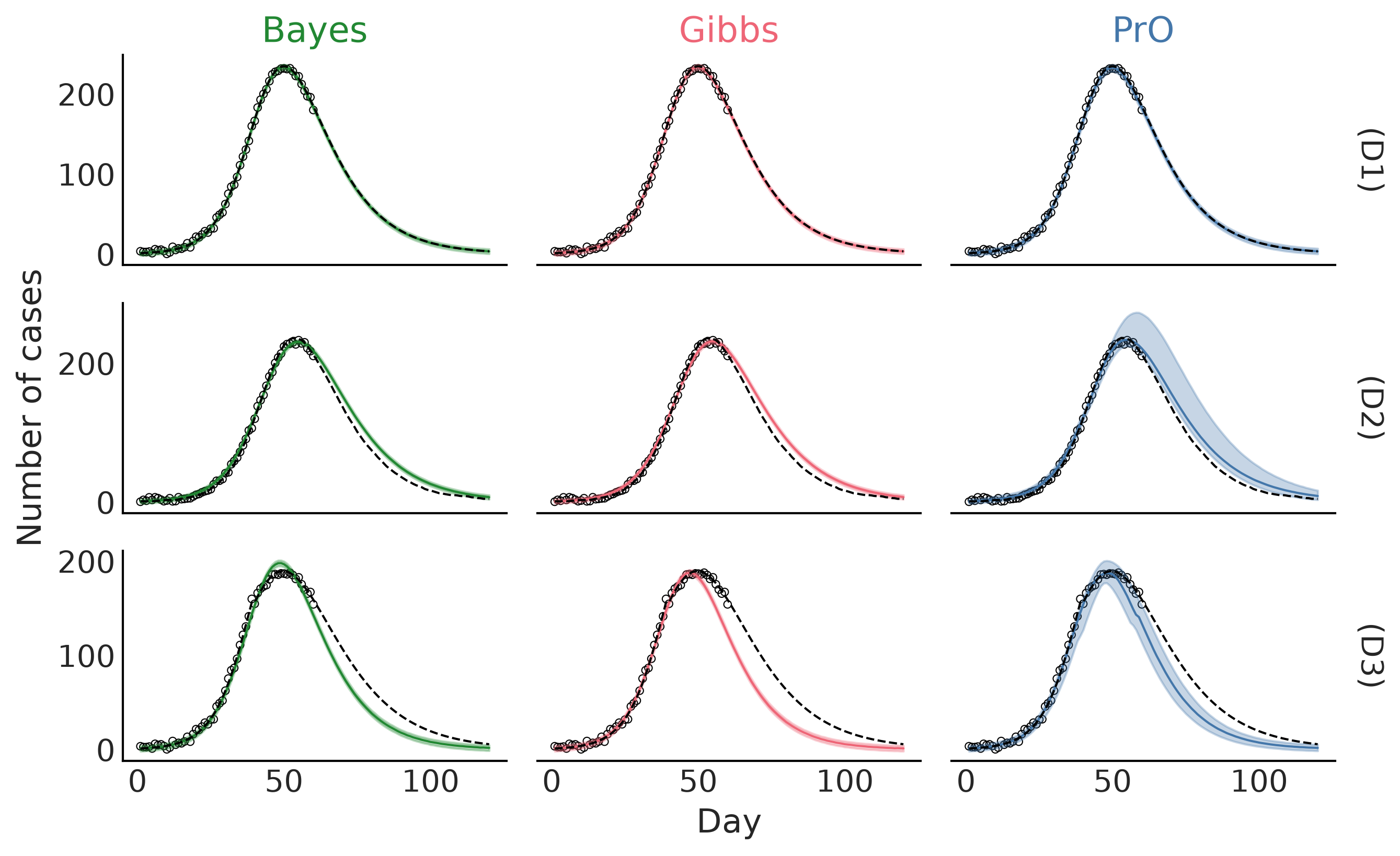}
    \caption{\textit{Simulated epidemiological data examples}\/.
    The dashed black line represents the mean of the data-generating process while the observed data are shown in hollow black circles.
    These are overlaid with the mean (solid lines) and $68\%$ credible region (ribbons) of the \textcolor{bayescolour}{\bfseries\sffamily Bayes}, \textcolor{gibbscolour}{\bfseries\sffamily Gibbs}, and \textcolor{procolour}{\bfseries\sffamily \ac{pro}} posterior predictive distributions.}
    \label{fig:sir-all-predictives}
\end{figure}
We first consider an example from simulation-based modelling:
when the simulator is misspecified the Bayes posterior concentrates and in turn its predictive can be grossly over-confident \citep{shen_prediction-centric_2025};
the \ac{pro} posterior by contrast expresses uncertainty in line with the discrepancy between the simulator and the true data-generating process.

The standard compartmental model in epidemiology characterises individuals as flowing between the states of susceptible ($S$), infected ($I$), and recovered ($R$) according to the ordinary differential equations
\begin{equation}\label{eq:sir-diff-eq}
    \frac{\dt S}{\dt t} = -\frac{\beta}{N} I S, \quad
    \frac{\dt I}{\dt t} = \frac{\beta}{N} I S - \gamma I, \quad
    \frac{\dt R}{\dt t} = \gamma I,
\end{equation}
which are parametrised by the \textit{infection rate} $\beta$ and the \textit{recovery rate} $\gamma$.
Their ratio $R_0 = \beta / \gamma$ is known as the \textit{reproduction number}, and characterises the infection dynamics.

Write $g_{t}(\theta;\,y_0) = (S_t,\,I_t,\,R_t) \in\mathbb{R}^3$ for the solution of~\cref{eq:sir-diff-eq} under $\theta = (\beta,\gamma)\in\mathbb{R}^2$ at time $t = 1,\,\ldots,T$, and given the initial conditions $y_0$ which are assumed to be fixed and known.
We consider three simulated data settings in which we inject additive Gaussian noise to these latent solutions, $y_t = g_t(\theta^\star; \,y_0) + \eta_t$, where $\eta_t \sim \normal(0,\,\sigma^2 \mathrm{I}_3)$ for $\sigma^2 = 3^2$, and for some $\theta^\star$:
\begin{enumerate}[label={\sffamily(D\arabic*)},ref={\sffamily D\arabic*},nosep,topsep=1em]
    \item 
    data are sampled according to~\cref{eq:sir-diff-eq} as described above;
    \label{dgp:well-specified}
    \item 
    data are sampled according to~\cref{eq:sir-diff-eq} with an additional Gaussian noise driving the infected compartment,
    \begin{equation*}
        \mathrm{d}I_t = \left(\frac{\beta}{N} I_t S_t - \gamma I_t\right)\mathrm{d}t + \varepsilon \mathrm{d} B_t,
    \end{equation*}
    where $\varepsilon = 0.4$ and $B=\{B_t\}_{t\geq 0}$ is a standard Brownian motion; and,\label{dgp:sde}
    \item data are sampled according to~\cref{eq:sir-diff-eq} with the underlying infection rate $\beta_t^\star$ varying around a change point $t_\mathrm{cp}$, so that $\beta_t^\star = 0.25$ for all $t < t_\mathrm{cp}$ and $\beta_t^\star = 0.2$ for all $t \geq t_\mathrm{cp}$.
    \label{dgp:lockdown}
\end{enumerate}
We fit the same model in each setting.
The model assumes a constant population so  that for some $N\in\mathbb{N}$, it holds that $N=S_t+I_t+R_t$ at all $t$.
In practice, we only observe the number of infected individuals, $I_t$, at each time step, and we fit these data to the observational model 
\begin{equation*}%
    p_\theta(I_{1:T}\mid y_0) = \prod_{t=1}^{T}\frac{1}{\sqrt{2\pi\hat \sigma^2}}\exp\left[-\frac{\{ I_t - \hat I_t(\theta;\,y_0)\}^{2}}{2\hat \sigma^{2}}\right],
\end{equation*}
where $\hat I_t(\theta;\,y_0)$ denotes the infected component of $g_t(\theta;\,y_0)$.
Under misspecification, $\sigma$ absorbs model discrepancy as well as observation noise \citep{brynjarsdottir_learning_2014};
we treat $\hat \sigma$ as fixed and set it to a plug-in estimate obtained from the residuals of a least-squares fit.
We compute three posteriors corresponding to this model:
the standard Bayes posterior~\cref{eq:bayes-posterior-defn}, a Gibbs posterior~\cref{eq:gibbs-posterior-defn} using the maximum mean discrepancy as $S_n$, and a particle \ac{pro} posterior~\cref{eq:particle-gibbs-density} also under the maximum mean discrepancy.
For all three, sampling is performed with Hamiltonian Monte Carlo and the \ac{pro} posterior was fit with $m=16$ particles. 

The posterior predictive distributions overlaid with observed data are shown in~\Cref{fig:sir-all-predictives}.
Naturally, this model is well-specified only in the~\ref{dgp:well-specified} regime. 
In this case, all three posterior predictive distributions converge on the true data-generating law.
In contrast, the model fails to account for the extra stochasticity induced in the~\ref{dgp:sde} regime. 
As a result, the Bayes and Gibbs posteriors overconfidently infer a predictive distribution which is a poor descriptor of the true data-generating process.
The \ac{pro} posterior delivers a more faithful quantification of the uncertainty in the data.
The~\ref{dgp:lockdown} regime  describes public health interventions such as lockdowns:
following an initially high reproduction rate,
the infection rate decreases after an intervention at time $t_{\operatorname{cp}}$.
This is not accounted for in the model and so the Bayes and Gibbs posteriors are concentrated around the single reproduction number best suited to describe \textit{all} of the observed data.
The \ac{pro} posterior however can produce a predictive distribution which is less confidently wrong.
\subsection{Spatial crime modelling}\label{sec:guerry}
\begin{figure}[t!]
    \centering
    \includegraphics[width=\linewidth]{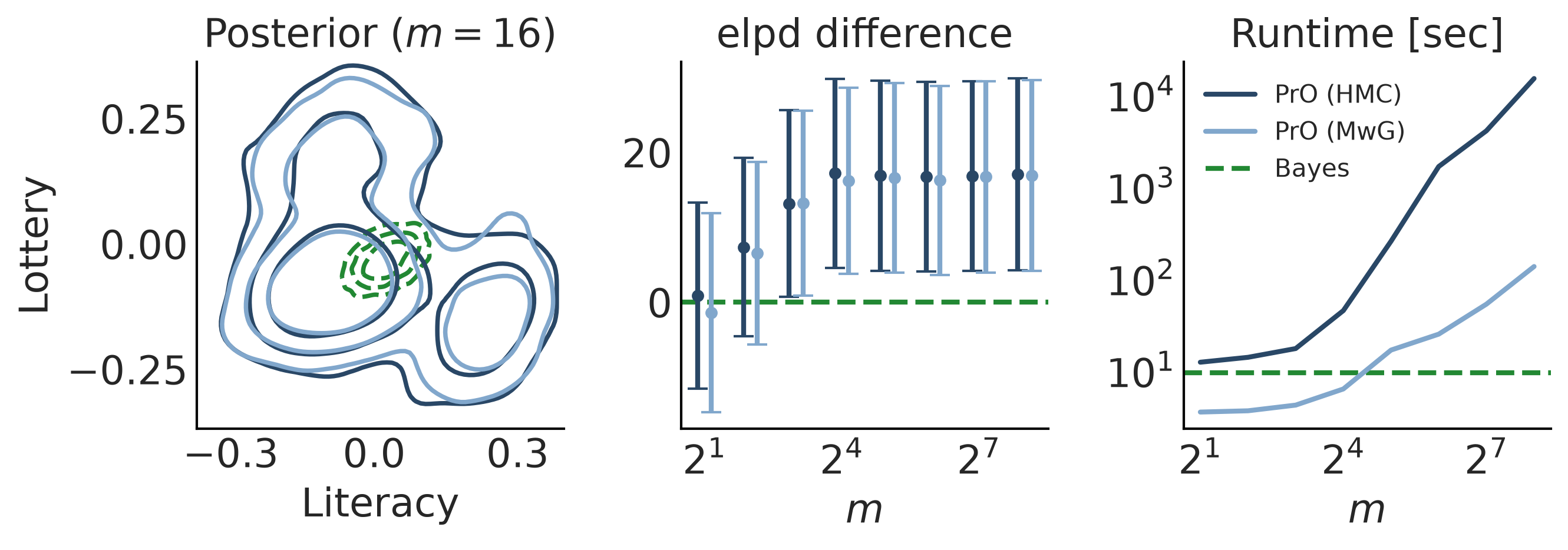}
    \caption{\textit{Moral statistics in 1830s France}\/. 
    The left-most facet shows the \textcolor{bayescolour}{\bfseries\sffamily Bayes} posterior and two \textcolor{procolour}{\bfseries\sffamily \ac{pro}} posteriors (fit with Hamiltonian Monte Carlo, HMC, and Metropolis-within-Gibbs, MwG, respectively) over two parameters.
    The middle facet shows the mean and standard error of the leave-one-out expected log predictive density (elpd) difference to the Bayes posterior of the \ac{pro} posteriors as a function of $m$:
    values greater than zero indicate better performance than the Bayes posterior.
    The right-most facet compares the average runtime to get $4{,}000$ draws across folds of the three posteriors as a function of $m$.}
    \label{fig:guerry}
\end{figure}
While our theoretical developments did not explicitly account for dependent data structures, we next show that the particle \ac{pro} posterior performs well on a spatial prediction example.
\cite{guerry_essai_1833} records statistics across France's $n=85$ regions.
For each region $i=1,\dots 85$, these include crime counts $y_i$ and seven socio-economic indicators $x_i$ such as literacy, wealth, and lottery participation. 
The population of the $i$-th region is denoted $\mathrm{pop}_i$.
Taking inspiration from \citet{friendly_guerry_2007}, we denote $y_{-i} = \sum_{j\neq i}y_j$ and model crime counts given the parameters $\theta = (\beta_0,\beta,b)\in\mathbb{R}^9$ as
\begin{equation*}
    p_\theta(y_i\mid x_i,\,y_{-i}) = \frac{\lambda_i^{y_i}e^{-\lambda_i}}{y_i!},\qquad 
    \log \lambda_i = \beta_0 + x_i^\top \beta + \log E_i + b\, \ell_i,
\end{equation*}
where $E_i = \mathrm{pop}_i( \sum_{j=1}^{n} y_j \big/ \sum_{j=1}^{n} \mathrm{pop}_j)$ is the expected count under population scaling, and acts as an offset term.
The spatial dependence enters through the lag term
$\ell_i = \sum_{j\ne i} \tilde{w}_{ij} y_j/E_j$ with $\tilde{w}_{ij} = w_{ij} / \sum_{k=1}^{n} w_{ik}$,
where $w_{ij}=1$ if regions $i$ and $j$ share a border or vertex, $0$ otherwise, and $w_{ii}=0$.
As a result, the coefficient $b$ measures the strength of spatial autocorrelation: 
$b>0$ means a region's risk rises with the average risk of its neighbours, and $b<0$ implies the opposite.
We infer both a standard Bayes posterior and a \ac{pro} posterior based on the log score.

The left-most panel of \Cref{fig:guerry} 
compares
the Bayes posterior over the coefficients associated with literacy and lottery to that of two \ac{pro} posteriors:
one fit with Hamiltonian Monte Carlo as implemented in the \texttt{blackjax} package \citep{cabezas_blackjax_2024}, and another with the blocked Gibbs sampler developed in~\Cref{sec:replica-gibbs}.
While the Bayes posterior is tightly concentrated, the \ac{pro} posteriors reflect a higher degree of uncertainty and are multimodal, with both samplers recovering similar multimodalities.
This more complete articulation of posterior uncertainty yields better predictions: 
the \ac{pro} posteriors attain significantly higher leave-one-region-out predictive log scores relative to the Bayes posterior as $m$ grows, as predicted by the results of~\Cref{sec:theory}.
In the middle panel, for instance, we find that already when $m = 16$ the \ac{pro} posteriors both perform significantly better than the Bayes posterior.
The right-most panel compares the runtime performance of the two samplers.
The results are encouraging:
the blocked Gibbs sampler provides an efficiency gain of two orders of magnitude.
In particular, while Hamiltonian Monte Carlo takes approximately four hours to draw $4{,}000$ samples at $m = 256$, the blocked Gibbs algorithm only takes around two minutes.
Even more notably still, at $m = 16$, the \ac{pro} posterior in the left-most panel of~\Cref{fig:guerry} fit with the blocked Gibbs sampler takes a similar amount of time as the standard Bayesian posterior.
\subsection{Low-energy nuclear physics}\label{sec:physics}
\begin{figure}
    \centering
    \includegraphics[width=\linewidth]{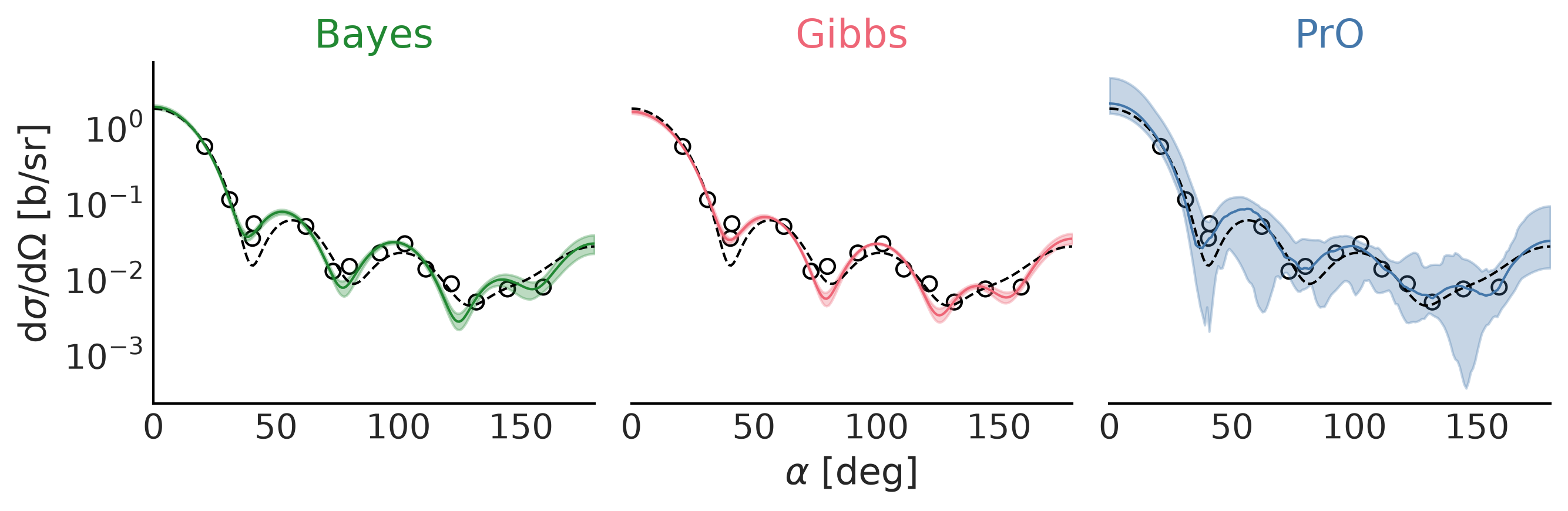}
    \caption{%
    \textit{Optical calibration of $\,^{40} \rm{Ca}(n,el)$ at 14.1 MeV.}
    In the dashed black line we show the best fit solution of \citet{koning_local_2003}  compared to the mean (solid lines) and $68\%$ credible region (ribbons) of the \textcolor{bayescolour}{\bfseries\sffamily Bayes}, \textcolor{gibbscolour}{\bfseries\sffamily Gibbs}, and \textcolor{procolour}{\bfseries\sffamily \ac{pro}} posterior predictive distributions.
    The observed data are shown in hollow black circles.}
    \label{fig:calibration-predictives}
\end{figure}
Finally, we consider a setting in which the \ac{wgf} is unavailable due to the likelihood being non-differentiable, and thus alternative sampling schemes are required.
This example concerns the interaction between a neutron projectile and the atomic nucleus it is shot into as part of a nuclear reaction.
Specifically, a neutron beam of incident energy $14.1\,\mathrm{MeV}$ is directed at a thin target of isotopically purified ${}^{40}\mathrm{Ca}$, and the angular distribution of scattered neutrons---also referred to as the differential cross section $\mathrm{d}\sigma/\mathrm{d}\Omega$---is recorded across a range of scattering angles \citep{mcdonald_scattering_1964}.
In principle, this scattering could be modelled exactly as a $41$-body problem \citep{hergert_guided_2020} but such a model is prohibitively expensive. 
Physicists instead reduce the reaction dynamics to an effective two-body problem through a parametric model called an optical potential \citep{thompson_nuclear_2009}.

For $i = 1,\ldots,n$ let $y_i$ denote the differential cross section $\mathrm{d}\sigma/\mathrm{d}\Omega$ at scattering angle $\alpha_i$ recorded by \citet{mcdonald_scattering_1964}.
Each measurement comes with a variance term $s_i^2 > 0$ which is estimated from the measurement equipment and treated as fixed throughout.
The interaction can be described by a Woods--Saxon optical potential \citep[Section 4.1.1]{thompson_nuclear_2009} with parameters $\theta = (V_v, W_v, R_v, a_v, W_d, R_d, a_d)$ \citep{koning_local_2003}. 
We write $g(\alpha_i;\theta)$ for the solution to the Schr\"odinger equation for the optical potential parametrised by $\theta$.
The observations are independent conditionally on their angles;
we follow the convention in the literature and adopt the conditionally Gaussian observational model \citep{beyer_bandframework_2025,beyer_JITR_2024}
\begin{equation*}
    p_\theta(y_{1:n}\mid \alpha_{1:n}) = \prod_{i=1}^{n} \frac{1}{\sqrt{2\pi s_i^{2}}} \exp\left[-\frac{\{y_i - g(\alpha_i;\theta)\}^{2}}{2 s_i^{2}}\right].
\end{equation*}
Two features of this model complicate inference. 
First, it is a deliberate approximation of the true data generating process and therefore clearly misspecified \citep[see also][]{phillips_BAND_2021}.
Second, the forward map $\theta\mapsto g(\cdot;\theta)$ is non-differentiable. 
This rules out use of standard gradient-based samplers such as interacting particle systems derived from a \ac{wgf}~\cref{eq:particle-dWFn} and gradient-based Monte Carlo samplers targetting~\cref{eq:particle-gibbs-density}.
Fortunately, gradient-free samplers such as random-walk Metropolis--Hastings are well established for targetting~\cref{eq:particle-gibbs-density} instead.
Further, the parameters $(R_v, a_v, R_d, a_d)\in\mathbb{R}_+^4$ are constrained to be positive. 
In order to evolve the \ac{wgf} corresponding to the exact \ac{pro} posterior in this case, one would need to define it with respect to a specialised geometry that respects these constraints \citep{gu_mirror_2026}.
A simple log-transform, on the other hand, makes this trivial for the random-walk Metropolis--Hastings sampler.

\begin{figure}[t!]
    \centering
    \includegraphics[width=\linewidth]{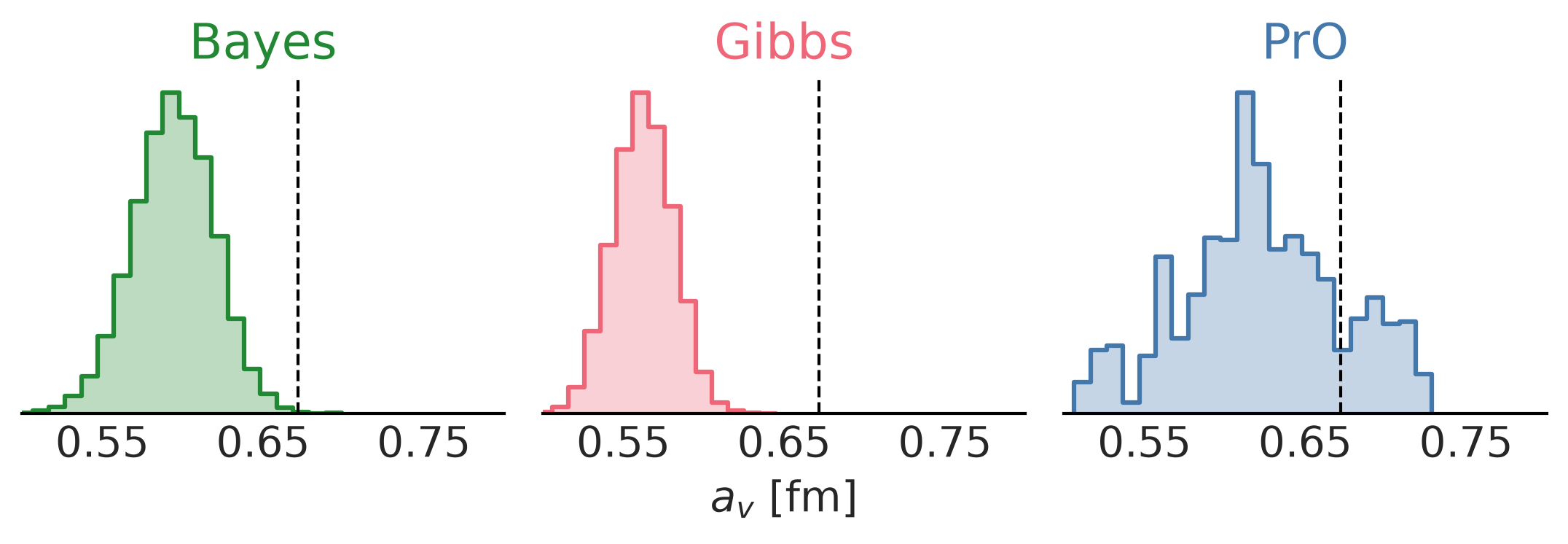}
    \caption{%
    \textit{Optical calibration of $\,^{40} \rm{Ca}(n,el)$ at 14.1 MeV.}
    The histograms represent draws from \textcolor{bayescolour}{\bfseries\sffamily Bayes}, \textcolor{gibbscolour}{\bfseries\sffamily Gibbs}, and \textcolor{procolour}{\bfseries\sffamily \ac{pro}} posterior distributions over the $a_v$ parameter.
    The dashed black lines indicate the point estimate of \citet{koning_local_2003}.}
    \label{fig:calibration-posteriors}
\end{figure}
Empirically, we find that both the Bayes and Gibbs predictive distributions are systematically overconfident, which mirrors the findings of \citet{pruitt_uncertainty-quantified_2023}.
In~\Cref{fig:calibration-predictives}, for instance, the Bayes posterior predictive's $68\%$ credible interval covers only half of the observations. 
Even where the predictive mean tracks the general shape of the data, the predictive spread almost certainly understates the uncertainty in the underlying nuclear reaction. 
Replacing the likelihood with the maximum mean discrepancy \citep{cherief-abdellatif_mmd-bayes_2019} is a natural response to misspecification:
bounded kernels make the posterior robust to extreme observations. 
But the optical potential does not describe all the data well at any parameter value, and the deficiency is systematic rather than due to outliers.
The same pathology as in the Bayes posterior is thus also present in the Gibbs posterior.
By contrast, the \ac{pro} posterior yields substantially more faithful predictive uncertainty. 
This improved predictive uncertainty is a product of the \ac{pro} posterior's parameter uncertainties, which happen to align much more closely with the parameter uncertainties as understood within the physics literature.
An instructive example for this is the diffuseness parameter $a_v$, which controls how sharp the optical potential well is. 
Roughly speaking, small $a_v$ describe a billiard ball-like nucleus with a sharp edge, while larger $a_v$ suggest a softer nuclear surface with less clearly defined edges. 
Historic estimates for $a_v$ for the reaction in question are often in the  range $0.6$--$0.8\,\mathrm{fm}$, which closely aligns with the parameter uncertainty communicated by the \ac{pro} posterior \citep{hodgson_nuclear_1971}.
In contrast, \Cref{fig:calibration-posteriors} illustrates that both Bayes and Gibbs posteriors concentrate near $0.6\,\mathrm{fm}$ but place substantial mass below the accepted range and almost none near the point estimate of \citet{koning_local_2003}, suggesting a sharper nuclear surface than is generally accepted in the literature.
\section{Discussion}\label{sec:discussion}
This paper transforms the \ac{pro} posterior from a theoretically grounded inferential object to a practically computable one.
Our results justify sampling from the \ac{pro} posterior using standard \acs{mcmc} methods, and constitute a step change in the practicality of the \ac{pro} framework.
While the observation that the particle discretisation of the mean-field Langevin dynamics possesses a tractable limiting measure that may be sampled directly is not specific to the \ac{pro} posterior, our generalisation guarantees are, and our finite-particle guarantees require much weaker conditions than classical propagation of chaos results.
Future work may look to establish heuristics for choosing $m$ based on a finite computational budget and desired accuracy.
Similarly, our theoretical analysis omits covariates;
if these can be addressed, there is also the opportunity to extend the blocked Gibbs sampler to generalised linear models \citep{chen_conjugate_2003}.

\subsection*{Acknowledgements}
We are grateful for enlightening discussions with Prof. Filomena Nunes, Dr. Kyle Beyer, Dr. Adrien Corenflos, and Dr. Peter A. Whalley which greatly improved this manuscript.
YM is supported by EP/V521917/1 and DTF by DP250101069.
\bibliography{main}
\newpage
\appendix
\counterwithin{figure}{section}
\counterwithin{assumption}{section}
\counterwithin{lemma}{section}
\counterwithin{theorem}{section}
\counterwithin{equation}{section}
\section{Proofs of main results}\label{sec:proofs}
\subsection{Proof of~\Cref{lem:beta-nonnegative}}\label{sec:non-negativity}
\begin{proof}
    For any measurable $A\subseteq\Theta$,
    \begin{equation*}
        \widehat Q_m(A) = \frac{1}{m}\sum_{j=1}^m \mathbf{1}(\theta_j\in A)
    \end{equation*}
    so that 
    \begin{IEEEeqnarray*}{rCl}
        \Ecurl[Q^{\otimes m}]{\widehat Q_m(A)} &=& \Ecurl[Q^{\otimes m}]{\frac{1}{m}\sum_{j=1}^m \mathbf{1}(\theta_j\in A)} \\
        &=& \frac{1}{m}\sum_{j=1}^m\Ecurl[Q^{\otimes m}]{\mathbf{1}(\theta_j\in A)} \\
        &=& \frac{1}{m}\sum_{j=1}^m\mathbb{P}_{Q^{\otimes m}}(\theta_j\in A) \\
        &=& \frac{1}{m}\sum_{j=1}^m\mathbb{P}_{Q}(\theta\in A) \\
        &=& \frac{1}{m}\sum_{j=1}^m Q(A) \\
        &=& Q(A),
    \end{IEEEeqnarray*}
    where the second line is by the linearity of expectation and the third by definition of expectations of indicator functions.
    Hence $\E[Q^{\otimes m}]{\widehat Q_m}=Q$, and since $P\mapsto S_n(P)$ is convex,
    \begin{equation*}
    S_n(P_Q)
    =
    S_n\left\{P_{\E[Q^{\otimes m}]{\widehat Q_m}}\right\}
    \le
    \Ecurl[Q^{\otimes m}]{S_n\bigl(P_{\widehat Q_m}\bigr)}
    =
    \int \Lnm(\theta_{1:m}) \dt Q^{\otimes m}(\theta_{1:m}).
    \end{equation*}
    This is exactly $\beta_{n,m}(Q)\ge 0$.
\end{proof}
\subsection{Proof of~\Cref{lemma:kl-closeness}}
In order to establish the stated result, we first present the following lemma.
\begin{lemma}\label{prop:entropy-curvature}
    Suppose that~\Cref{ass:convex} and the conditions of~\Cref{lemma:existence-uniqueness-full-pro} hold, and that the first variation in~\cref{eq:pro-self-gibbs} is finite $\Pi$-almost surely.
    Define 
    \begin{equation*}
        g_n(\theta)
        =
        \left.\frac{1}{n}\sum_{i=1}^n\frac{\delta S(P_Q, x_i)}{\delta Q}\right|_{Q=\Qn}(\theta),
    \end{equation*}
    and
    \begin{equation}\label{eq:pro-score-bregman}
        \mathcal{B}_n(Q,\Qn)
        =
        S_n(P_Q)-S_n(P_{\Qn})
        -
        \int_\Theta g_n(\theta)\dt(Q-\Qn)(\theta).
    \end{equation}
    Then $\mathcal{B}_n(Q,\Qn)\geq0$, and for every $Q\in\mathcal{P}_\Pi(\Theta)$ for which the terms below are finite, 
    \begin{multline}\label{eq:pro-entropy-curvature}
        \left\{
        S_n(P_Q)
        +
        \frac{\kld{Q}{\Pi}}{\lambda_n}
        \right\}
        -
        \left\{
        S_n(P_{\Qn})
        +
        \frac{\kld{\Qn}{\Pi}}{\lambda_n}
        \right\}
        =
        \mathcal{B}_n(Q,\Qn)
        +
        \frac{\kld{Q}{\Qn}}{\lambda_n}.
    \end{multline}
\end{lemma}
\begin{proof}
        By~\cref{eq:pro-self-gibbs}, there exists a constant \(c_n\in\mathbb{R}\) such that, for $\Pi$-almost every $\theta$, 
    \begin{equation}\label{eq:pro-fixed-point-log-density}
        \log\frac{\dt\Qn}{\dt\Pi}(\theta)
        =
        -\lambda_n g_n(\theta)-c_n.
    \end{equation}
    Since \(g_n\) is finite \(\Pi\)-almost surely, \(\Qn\) and \(\Pi\) are mutually absolutely continuous, and hence every \(Q\in\mathcal{P}_\Pi(\Theta)\) is absolutely continuous with respect to \(\Qn\).
    By direct computation, we then have
    \begin{IEEEeqnarray*}{rCl}
        \kld{Q}{\Pi}-\kld{\Qn}{\Pi}
        &=&
        \kld{Q}{\Qn}
        +
        \int_\Theta
        \log\frac{\dt\Qn}{\dt\Pi}(\theta)
        \dt(Q-\Qn)(\theta)
        \\
        &=&
        \kld{Q}{\Qn}
        -
        \lambda_n
        \int_\Theta g_n(\theta)\dt(Q-\Qn)(\theta),
    \end{IEEEeqnarray*}
    where the constant $c_n$ vanishes because $Q$ and $\Qn$ are probability measures.
    Dividing both sides by $\lambda_n$ and adding $S_n(P_Q)-S_n(P_{\Qn})$ establishes~\cref{eq:pro-entropy-curvature}.
    Since the map $Q\mapsto P_Q$ is linear and $P\mapsto S_n(P)$ is convex under~\Cref{ass:convex}, the map $Q\mapsto S_n(P_Q)$ is convex.
    The subgradient inequality \citep[Lemma 4.1]{hu_mean-field_2021} applied at $\Qn$ therefore yields $\mathcal{B}_n(Q,\Qn)\geq0$.
\end{proof}
We also require the additional lemma below, which allows us to compare the objective of~\cref{eq:pro-posterior-defn} induced by the common finite-particle marginal with that of any posterior predictive distribution.
\begin{lemma}\label{lemma:bound}
    Let $\mathcal{P}_\Pi(\Theta)=\{Q\in\mathcal{P}(\Theta):Q\ll\Pi\}$.
    Under~\Cref{ass:convex}, for a prior $\Pi\in\mathcal{P}(\Theta)$ and any $Q\in\mathcal{P}_\Pi(\Theta)$ such that $Q^{\otimes m}\ll \Pi^{\otimes m}$,
    \begin{equation*}
        S_{n}(P_{\barQnm}) + \frac{\kld{\barQnm}{\Pi}}{\lambda_n} \leq S_n(P_{Q}) + \frac{\kld{Q}{\Pi}}{\lambda_n} + \beta_{n,m}(Q).
    \end{equation*}
    In particular, 
    \begin{equation*}
        0\leq \left\{S_{n}(P_{\barQnm}) + \frac{\kld{\barQnm}{\Pi}}{\lambda_n}\right\} - \left\{ S_n(P_{\Qn}) + \frac{\kld{\Qn}{\Pi}}{\lambda_n} \right\} \leq \beta_{n,m}(\Qn).
    \end{equation*}
\end{lemma}
\begin{proof}
    Recall that $\barQm = \int \widehat Q_m(\theta_{1:m}) \dt Q_m(\theta_{1:m})$, then by the convexity of $P\mapsto S_n(P)$ for all $Q_m\in\mathcal{P}(\Theta^m)$,
    \begin{IEEEeqnarray}{rCl}
        S_n(P_{\barQm}) &=& S_n\left\{P_{\int \widehat Q_m(\theta_{1:m}) \dt Q_m(\theta_{1:m})}\right\} \nonumber \\
        &\leq& \int  S_n\left\{P_{\widehat Q_m(\theta_{1:m})}\right\} \dt Q_m(\theta_{1:m}) \nonumber \\
        &=& \int \Lnm(\theta_{1:m}) \dt Q_m(\theta_{1:m}), \label{eq:Sn-bound}
    \end{IEEEeqnarray}
    where the third line follows by the definition of the finite-particle predictive score.
    Elsewhere, we have by standard information processing results that
    \begin{IEEEeqnarray}{rCl}
        \kld{Q_m}{\Pi^{\otimes m}} &=& \kldcurl{Q_m}{\bigotimes_{j=1}^m Q_m^{(j)}} + \sum_{j=1}^m \kldcurl{Q_m^{(j)}}{\Pi} \nonumber \\
        &\geq& \sum_{j=1}^m \kldcurl{Q_m^{(j)}}{\Pi} \nonumber \\
        &=& \frac{m}{m} \sum_{j=1}^m \kldcurl{Q_m^{(j)}}{\Pi} \nonumber \\
        &\geq& m\, \kldcurl{\frac{1}{m}\sum_{j=1}^m Q_m^{(j)}}{\Pi} \nonumber \\
        &=& m\, \kld{\barQm}{\Pi},\label{eq:kl-bound}
    \end{IEEEeqnarray}
    where the second line follows from the non-negativity of the Kullback--Leibler divergence, the third multiplies by one, the fourth by the convexity of the Kullback--Leibler divergence in its first argument, and the final line by definition of $\barQm$.
    As such, we can combine~\cref{eq:Sn-bound,eq:kl-bound} to write 
    \begin{equation*}
        S_n(P_{\barQm}) + \frac{\kld{\barQm}{\Pi}}{\lambda_n} \leq \int \Lnm(\theta_{1:m}) \dt Q_m(\theta_{1:m}) + \frac{\kld{Q_m}{\Pi^{\otimes m}}}{m\cdot\lambda_n}.
    \end{equation*}
    And since this holds for any $Q_m\in\mathcal{P}(\Theta^m)$, it holds for $\Qnm$:
    the minimiser of the right hand side:
    \begin{IEEEeqnarray*}{rCl}
        S_n(P_{\barQnm}) + \frac{\kld{\barQnm}{\Pi}}{\lambda_n} &\leq& \int \Lnm(\theta_{1:m}) \dt \Qnm(\theta_{1:m}) + \frac{\kld{\Qnm}{\Pi^{\otimes m}}}{m\cdot\lambda_n} \\
        &=& \inf_{Q_m\in\mathcal{P}(\Theta^m)} \left[\int \Lnm(\theta_{1:m}) \dt Q_m(\theta_{1:m}) + \frac{\kld{Q_m}{\Pi^{\otimes m}}}{m\cdot\lambda_n}\right] \\
        &\leq& \int \Lnm(\theta_{1:m}) \dt Q^{\otimes m}(\theta_{1:m}) + \frac{\kld{Q^{\otimes m}}{\Pi^{\otimes m}}}{m\cdot\lambda_n} \\
        &=& \int \Lnm(\theta_{1:m}) \dt Q^{\otimes m}(\theta_{1:m}) + \frac{\kld{Q}{\Pi}}{\lambda_n},
    \end{IEEEeqnarray*}
    where the third line now applies an arbitrary product measure constructed from $Q\in\mathcal{P}(\Theta)$ and the fourth line follows from the properties of the Kullback--Leibler divergence for product measures.
    Finally, apply the definition of $\beta_{n,m}(Q)$ to attain
    \begin{equation*}
        S_n(P_{\barQnm}) + \frac{\kld{\barQnm}{\Pi}}{\lambda_n} \leq S_n(P_{Q}) + \frac{\kld{Q}{\Pi}}{\lambda_n} + \beta_{n,m}(Q),
    \end{equation*}
    as desired.
\end{proof}
\begin{proof}[of~\Cref{lemma:kl-closeness}]
    Applying~\Cref{prop:entropy-curvature} with $Q=\barQnm$, using the non-negativity of $\mathcal{B}_n(\barQnm,\Qn)$, and then applying~\Cref{lemma:bound} gives
    \begin{IEEEeqnarray*}{rCl}
        \frac{\kld{\barQnm}{\Qn}}{\lambda_n}
        &\leq&
        \left\{
        S_n(P_{\barQnm})
        +
        \frac{\kld{\barQnm}{\Pi}}{\lambda_n}
        \right\}
        -
        \left\{
        S_n(P_{\Qn})
        +
        \frac{\kld{\Qn}{\Pi}}{\lambda_n}
        \right\}\\
        &\leq&
        \beta_{n,m}(\Qn).
    \end{IEEEeqnarray*}
    This proves the relative-entropy bound. 
    The final display follows immediately from~\Cref{ass:non-trivial-misspecification-exact} since $\beta_{n,m}(\Qn)\lesssim 1/m$ for a fixed $n$.
\end{proof}
    An application of Pinsker's inequality provides a version of the stated result in total variation distance:
    \begin{equation*}
        \tvd{\barQnm}{\Qn}
        \leq
        \sqrt{\frac{\lambda_n C_n}{2m}}.
    \end{equation*}
    That is, the ideal \ac{pro} posterior and the marginal of the finite-particle \ac{pro} posterior are asymptotically indistinguishable insofar as integrating bounded continuous functions.
\subsection{Proof of~\Cref{thm:predictive-concentration}}
\begin{proof}
    By~\Cref{lemma:bound}, after dropping the non-negative Kullback--Leibler term of the left-hand side, and applying $Q^\star = \argmin_{Q\in\mathcal{P}(\Theta)} \mathcal{S}(P_Q,\,P_0)$,
    \begin{equation*}
        S_n(P_{\barQnm}) \leq S_n(P_{Q^\star}) + \frac{\kld{Q^\star}{\Pi}}{\lambda_n} + \beta_{n,m}(Q^\star).
    \end{equation*}
    Write 
    \begin{equation*}
        \Delta_n = \sup_{Q\in\mathcal{P}_\Pi} \vert S_n(P_Q) - \mathcal{S}(P_Q,\,P_0)|,
    \end{equation*}
    so that
    \begin{equation*}
        \mathcal{S}(P_{\barQnm},\,P_0) \leq S_n(P_{Q^\star}) + \frac{\kld{Q^\star}{\Pi}}{\lambda_n} + \beta_{n,m}(Q^\star) + \Delta_n.
    \end{equation*}
    Setting $\lambda_n \asymp n^{1/2} / \log(n)$ and taking expectations of both sides with respect to $P_0$ yields
    \begin{equation*}
        \Ecurl[0]{\mathcal{S}(P_{\barQnm},\,P_0)} \lesssim \Ecurl[0]{S_n(P_{Q^\star})} + \frac{\log(n)}{n^{1/2}} + \Ecurl[0]{\beta_{n,m}(Q^\star)} + \E[0]{\Delta_n}.
    \end{equation*}
    Under~\Cref{ass:non-trivial-misspecification-exact,ass:concentration_existence}, $\Ecurl[0]{\beta_{n,m}(Q^\star)}\lesssim m^{-1}$, so that by the definition of $\mathcal{S}(\cdot,\,P_0)$ and~\Cref{ass:entropy}
    \begin{IEEEeqnarray*}{rCl}
        \Ecurl[0]{\mathcal{S}(P_{\barQnm},\,P_0)} &\lesssim& \mathcal{S}(P_{Q^\star},\,P_0) + \frac{\log(n)}{n^{1/2}} + \frac{1}{m} + r_n \\
        &\lesssim& \mathcal{S}(P_{Q^\star},\,P_0) + \frac{1}{m} + \max\left\{\frac{\log(n)}{n^{1/2}}, r_n \right\}.
    \end{IEEEeqnarray*}
    Subtracting $\mathcal{S}(P_0,\,P_0)$ from both sides and applying the definition of $d_\mathcal{S}$ completes the proof alongside the definition of $Q^\star$. 
\end{proof}
\subsection{Proof of~\Cref{thm:predictive-concentration-kl}}
\begin{proof}
    The proof is inspired by the argument in Theorem 2.1 of \cite{alquier2020concentration}, where we use density notation throughout to simplify the derivations.
    Radon--Nikodym derivatives can be used with virtually no change. 
    First define
    \begin{equation*}
        \overline{P}_{\theta_{1:m}}=\frac{1}{m}\sum_{j=1}^{m} P_{\theta_j}=P_{\widehat{Q}(\theta_{1:m})}.
    \end{equation*}
    Then, we can write
    \begin{IEEEeqnarray*}{rl}
        \EE_0 & \left[e^{-\alpha \sum_{i=1}^{n}\log \left\{p_0(x_i) / \overline{p}_{\theta_{1:m}}(x_i)\right\}} \right]\\
        &\qquad=\EE_0 \left[\prod_{i=1}^{n} e^{-\alpha \log \left\{p_0(x_i) / \overline{p}_{\theta_{1:m}}(x_i)\right\}} \right] \\
        &\qquad=\EE_0 \left(\prod_{i=1}^{n} \exp \log \left[\left\{\frac{p_0(x_i)}{ \overline{p}_{\theta_{1:m}}(x_i)}\right\}^{-\alpha}\right] \right) \\
        &\qquad=\EE_0 \left(\prod_{i=1}^{n} \exp \log \left[\left\{\frac{ \overline{p}_{\theta_{1:m}}(x_i)}{p_0(x_i)}\right\}^{\alpha}\right] \right) \\
        &\qquad=\idotsint \left\{\frac{\overline{p}_{\theta_{1:m}}(x_1)}{p_0(x_1)}\right\}^\alpha\cdots\left\{\frac{\overline{p}_{\theta_{1:m}}(x_n)}{p_0(x_n)}\right\}^\alpha p_0(x_1)\cdots p_0(x_n) \dt x_1\cdots \dt x_n\\
        &\qquad=\int\cdots\int \overline{p}_{\theta_{1:m}}^\alpha(x_1)p_0(x_1)^{1-\alpha}\cdots \overline{p}_{\theta_{1:m}}^\alpha(x_n)p_0(x_n)^{1-\alpha}  \dt x_1\cdots \dt x_n \\
        &\qquad=\exp\bigg[\left\{\frac{\alpha - 1}{\alpha - 1}\right\}\\
        &\qquad\qquad \times \log\left\{\int\cdots\int \overline{p}_{\theta_{1:m}}^\alpha(x_1)p_0(x_1)^{1-\alpha}\cdots \overline{p}_{\theta_{1:m}}^\alpha(x_n)p_0(x_n)^{1-\alpha}  \dt x_1\cdots \dt x_n\right\}\bigg] \\
        &\qquad=\exp\{-(1-\alpha) d_\alpha(\overline{P}_{\theta_{1:m}}^{\otimes n};\,P_0^{\otimes n})\} \\
        &\qquad=\exp\{-n(1-\alpha) d_\alpha(\overline{P}_{\theta_{1:m}};\,P_0)\},
    \end{IEEEeqnarray*}
    where $d_\alpha$ is the R\'enyi $\alpha$-divergence as defined in~\cref{eq:renyi}.
    Hence, we have that 
    \begin{equation*}%
        \Esqr[0]{e^{-\alpha \sum_{i=1}^{n}\log \left\{p_0(x_i) / \overline{p}_{\theta_{1:m}}(x_i)\right\}}}=e^{- n(1-\alpha) d_\alpha(\overline{P}_{\theta_{1:m}};\,P_0)} .
    \end{equation*}
    Now, for any given value of $\theta_{1:m}$, it may be that $d_\alpha(\overline{P}_{\theta_{1:m}};\,P_0)=\infty$.
    Decompose the prior into
    \begin{IEEEeqnarray*}{rCl}          
        \dt\Pi^{\otimes m}(\theta_{1:m}) &=& \dt \Pi_1^{\otimes m}(\theta_{1:m})+\dt\Pi_2^{\otimes m}(\theta_{1:m})\\
        &=&\dt\Pi^{\otimes m}(\theta_{1:m})\mathbf{1}_{\{d_\alpha(\overline{P}_{\theta_{1:m}};\,P_0)<\infty\}}+\dt\Pi^{\otimes m}(\theta_{1:m})\mathbf{1}_{\{d_\alpha(\overline{P}_{\theta_{1:m}};\,P_0)=\infty\}}.
    \end{IEEEeqnarray*}
    If $\Pi=\Pi_2$, then $\Pi$-almost surely we have $\left\{d_\alpha(\overline{P}_{\theta_{1:m}};\,P_0)=\infty\right\}$ and $\sum_{i=1}^{n}\log \left\{p_0(x_i) / \overline{p}_{\theta_{1:m}}(x_i)\right\}=\infty$, which must then also hold $Q$-almost surely for any $Q\ll\Pi^{\otimes m}$. 
    In this case, the result is vacuously satisfied. 
    Assume now that $\Pi\ne\Pi_2$. 
    Redefine the prior as the renormalised version of $\Pi_1^{\otimes m}$ for all measurable $A\subseteq\Theta^m$.
    \begin{equation*}
        \widetilde\Pi_m(A) = \frac{\Pi_1^{\otimes m}(A)}{\Pi_1^{\otimes m}\left\{\theta_{1:m}\in\Theta^m:\,d_\alpha(\overline{P}_{\theta_{1:m}};\,P_0)<\infty\right\}}
    \end{equation*}  
    On the support of our renormalised prior ${\widetilde\Pi_m}$, we have 
    \begin{equation}\label{eq:equality}
        1=\Esqr[0]{e^{-\alpha \sum_{i=1}^{n}\log \left\{p_0(x_i) / \overline{p}_{\theta_{1:m}}(x_i)\right\}+ n(1-\alpha) d_\alpha(\overline{P}_{\theta_{1:m}};\,P_0)}} .
    \end{equation}
    For $\theta_{1:m}$ in the support of ${\widetilde\Pi_m}$, define
    \begin{equation*}
        r_{n,m}(\theta_{1:m}) = -\alpha \sum_{i=1}^{n}\log \left\{p_0(x_i) / \overline{p}_{\theta_{1:m}}(x_i)\right\} + n(1-\alpha) d_\alpha(\overline{P}_{\theta_{1:m}};\,P_0)
    \end{equation*}
    and the apply the Donsker--Varadhan lemma to it so that
    \begin{equation}\label{eq:dv-renyi}
        \log \int e^{r_{n,m}(\theta_{1:m})}\dt{\widetilde\Pi_m}(\theta_{1:m}) = \sup_{Q_m\in\mathcal{P}(\Theta^m)}\left\{\int r_{n,m}(\theta_{1:m})\dt Q_m(\theta_{1:m})-\kld{Q_m}{{\widetilde\Pi_m}}\right\}.
    \end{equation}
    Integrating~\cref{eq:equality} with respect to the (normalised) prior, applying Fubini's theorem, and substituting~\cref{eq:dv-renyi} into it leads to
    \begin{IEEEeqnarray*}{rCl}
        1 &=& \int \Esqr[0]{e^{-\alpha \sum_{i=1}^{n}\log \left\{p_0(x_i) / \overline{p}_{\theta_{1:m}}(x_i)\right\} + n(1-\alpha) d_\alpha(\overline{P}_{\theta_{1:m}};\,P_0)}} \dt{\widetilde\Pi_m}(\theta_{1:m}) \\
        &=& \Esqr[0]{\int e^{-\alpha \sum_{i=1}^{n}\log \left\{p_0(x_i) / \overline{p}_{\theta_{1:m}}(x_i)\right\} + n(1-\alpha) d_\alpha(\overline{P}_{\theta_{1:m}};\,P_0)} \dt{\widetilde\Pi_m}(\theta_{1:m})} \\
        &=& \EE_0 \exp \Bigg (\sup_{Q_m\in\mathcal{P}(\Theta^m)}\bigg[ -\alpha \int \sum_{i=1}^{n}\log \left\{\frac{p_0(x_i)}{\overline{p}_{\theta_{1:m}}(x_i)}\right\} \dt Q_m(\theta_{1:m}) \\
        && \qquad\qquad  + n(1-\alpha) \int d_\alpha(\overline{P}_{\theta_{1:m}};\,P_0) \dt Q_m(\theta_{1:m}) - \kld{Q_m}{{\widetilde\Pi_m}} \bigg] \Bigg)
    \end{IEEEeqnarray*}
    Now, take logs and apply Jensen's inequality to obtain
    \begin{multline*}
        0 \geq \EE_0\sup_{Q_m\in\mathcal{P}(\Theta^m)}\bigg[ -\alpha \int \sum_{i=1}^{n}\log \left\{\frac{p_0(x_i)}{\overline{p}_{\theta_{1:m}}(x_i)}\right\} \dt Q_m(\theta_{1:m}) \\ + n(1-\alpha) \int d_\alpha(\overline{P}_{\theta_{1:m}};\,P_0) \dt Q_m(\theta_{1:m}) - \kld{Q_m}{{\widetilde\Pi_m}} \bigg].
    \end{multline*}
    In particular, for $Q_m = \Qnm^{(\alpha)}$ we have
    \begin{multline*}
        0 \geq \EE_0\bigg[ -\alpha \int \sum_{i=1}^{n}\log \left\{\frac{p_0(x_i)}{\overline{p}_{\theta_{1:m}}(x_i)}\right\} \dt \Qnm^{(\alpha)}(\theta_{1:m}) \\ + n(1-\alpha) \int d_\alpha(\overline{P}_{\theta_{1:m}};\,P_0) \dt \Qnm^{(\alpha)}(\theta_{1:m}) - \kldcurl{\Qnm^{(\alpha)}}{{\widetilde\Pi_m}} \bigg].
    \end{multline*}
    Re-arranging this we have 
    \begin{IEEEeqnarray*}{rl}
        n(1-\alpha)\,&\EE_0\bigg\{\int d_\alpha(\overline{P}_{\theta_{1:m}};\,P_0)\dt \Qnm^{(\alpha)}(\theta_{1:m})\bigg\} \\
        &\leq \EE_0\left[\alpha\int \sum_{i=1}^{n}\log \left\{\frac{p_0(x_i)}{\overline{p}_{\theta_{1:m}}(x_i)}\right\}\dt\Qnm^{(\alpha)}(\theta_{1:m})+{\kldcurl{\Qnm^{(\alpha)}}{{\widetilde\Pi_m}}}\right]\\
        &=\EE_0\Bigg(\inf_{Q_m\in\mathcal{P}(\Theta^m)}\bigg[\alpha\int \sum_{i=1}^{n}\log \left\{\frac{p_0(x_i)}{\overline{p}_{\theta_{1:m}}(x_i)}\right\}\dt Q_m(\theta_{1:m})+{\kldcurl{Q_m}{{\widetilde\Pi_m}}}\bigg]\Bigg)
    \end{IEEEeqnarray*}
    where the equality uses the definition of $\Qnm^{(\alpha)}$.
    Applying Jensen's inequality to the right hand side yields
    \begin{multline}\label{eq:renyi-jensen}
        n(1-\alpha)\,\EE_0\bigg\{\int d_\alpha(\overline{P}_{\theta_{1:m}};\,P_0)\dt \Qnm^{(\alpha)}(\theta_{1:m})\bigg\} \\
        \le \inf_{Q_m\in\mathcal{P}(\Theta^m)}\left(\alpha\int \sum_{i=1}^{n} \Esqr[0]{\log \left\{\frac{p_0(x_i)}{\overline{p}_{\theta_{1:m}}(x_i)}\right\}}\dt Q_m(\theta_{1:m})+{\kld{Q_m}{{\widetilde\Pi_m}}}\right).
    \end{multline}
    By the definition of the renormalised prior $\widetilde\Pi_m$ 
    \begin{IEEEeqnarray}{rCl}
        \kld{Q^{\otimes m}}{\widetilde\Pi_m}&=& \kld{Q^{\otimes m}}{\Pi_1^{\otimes m}}+\log \Pi^{\otimes m}_1\left\{\theta_{1:m}\in\Theta^m:\,d_\alpha(\overline{P}_{\theta_{1:m}};\,P_0)<\infty\right\} \nonumber \\
        &\leq& \kld{Q^{\otimes m}}{\Pi^{\otimes m}} 
        \label{eq:restricted-prior-factorised}
    \end{IEEEeqnarray}
    Thus, evaluating~\Cref{eq:renyi-jensen} at $Q_m={Q^\star}^{\otimes m}$, we find
    \begin{IEEEeqnarray*}{rl}
        \EE_0\bigg\{\int& d_\alpha(\overline{P}_{\theta_{1:m}};\,P_0)\dt \Qnm^{(\alpha)}(\theta_{1:m})\bigg\} \\
        &\leq \frac{\alpha}{n(1-\alpha)} \int \sum_{i=1}^{n} \Esqr[0]{\log \left\{\frac{p_0(x_i)}{\overline{p}_{\theta_{1:m}}(x_i)}\right\}}\dt {Q^\star}^{\otimes m}(\theta_{1:m})+\frac{\kld{{Q^\star}^{\otimes m}}{{\widetilde\Pi_m}}}{n(1-\alpha)} \\
        &\le \frac{\alpha}{n(1-\alpha)} \int \sum_{i=1}^{n} \Esqr[0]{\log \left\{\frac{p_0(x_i)}{\overline{p}_{\theta_{1:m}}(x_i)}\right\}}\dt {Q^\star}^{\otimes m}(\theta_{1:m})+\frac{\kld{{Q^\star}^{\otimes m}}{{\Pi^{\otimes m}}}}{n(1-\alpha)} \\
        &= \frac{\alpha}{n(1-\alpha)} \int \sum_{i=1}^{n} \Esqr[0]{\log \left\{\frac{p_0(x_i)}{\overline{p}_{\theta_{1:m}}(x_i)}\right\}}\dt {Q^\star}^{\otimes m}(\theta_{1:m})+\frac{m\,\kld{{Q^\star}}{\Pi}}{n(1-\alpha)} \\
        &= \frac{\alpha}{1-\alpha} \int \Esqr[0]{\log \left\{\frac{p_0(X)}{\overline{p}_{\theta_{1:m}}(X)}\right\}}\dt {Q^\star}^{\otimes m}(\theta_{1:m})+\frac{m\,\kld{{Q^\star}}{\Pi}}{n(1-\alpha)}\\
        &= \frac{\alpha}{1-\alpha} \int \EE_0\bigg[\log \left\{\frac{p_0(X)}{\overline{p}_{\theta_{1:m}}(X)}\right\} - \log \left\{\frac{p_0(X)}{p_{Q^\star}(X)}\right\}\bigg]\dt {Q^\star}^{\otimes m}(\theta_{1:m}) \\
        &\qquad + \frac{\alpha}{1-\alpha}\int \bigg[\log \left\{\frac{p_0(X)}{p_{Q^\star}(X)}\right\} \bigg]\dt {Q^\star}^{\otimes m}(\theta_{1:m}) +\frac{m\,\kld{{Q^\star}}{\Pi}}{n(1-\alpha)} \\
        &= \frac{\alpha}{1-\alpha} \int \EE_0\bigg[\log \left\{\frac{p_0(X)}{\overline{p}_{\theta_{1:m}}(X)}\right\} - \log \left\{\frac{p_0(X)}{p_{Q^\star}(X)}\right\}\bigg]\dt {Q^\star}^{\otimes m}(\theta_{1:m}) \\
        &\qquad + \frac{\alpha}{1-\alpha}\kld{P_0}{P_{Q^\star}} +\frac{m\,\kld{{Q^\star}}{\Pi}}{n(1-\alpha)} \\
        &= \frac{\alpha}{1-\alpha} \Ecurl[0]{\beta_{n,m}(Q^\star)} + \frac{\alpha}{1-\alpha}\kld{P_0}{P_{Q^\star}} +\frac{m\,\kld{{Q^\star}}{\Pi}}{n(1-\alpha)},
    \end{IEEEeqnarray*}
    where the first line divides both sides by $n(1 - \alpha)$, the second follows from the fact that $\kld{{Q^\star}^{\otimes m}}{\widetilde\Pi_m{}}\le \kld{{Q^\star}^{\otimes m}}{\Pi^{\otimes m}}$ by~\cref{eq:restricted-prior-factorised}, the third line uses the properties of the Kullback--Leibler divergence, the fourth collapses the sum, the fifth adds zero, the sixth applies the definition of the Kullback--Leibler between $P_0$ and $P_{Q^\star}$, and the last line follows by the definition of $\beta_{n,m}(Q^\star)$.
    Under~\Cref{ass:log_score}, by \Cref{lemma:log-score-beta}, we have that $\Ecurl[0]{\beta_{n,m}(Q^\star)} \lesssim 1 / m$  so that
    \begin{IEEEeqnarray*}{rl}
        \EE_0\bigg\{\int& d_\alpha(\overline{P}_{\theta_{1:m}};\,P_0)\dt \Qnm^{(\alpha)}(\theta_{1:m})\bigg\} \\
        &\leq \frac{\alpha}{1-\alpha} \Ecurl[0]{\beta_{n,m}(Q^\star)} + \frac{\alpha}{1-\alpha}\kld{P_0}{P_{Q^\star}} +\frac{m\,\kld{{Q^\star}}{\Pi}}{n(1-\alpha)} \\
        &\lesssim \frac{\alpha}{m(1-\alpha)} + \frac{\alpha}{1-\alpha}\kld{P_0}{P_{Q^\star}} +\frac{m\,\kld{{Q^\star}}{\Pi}}{n(1-\alpha)}.
    \end{IEEEeqnarray*}
    To obtain the stated result, recall that the R\'enyi divergence is convex in its first argument for $\alpha\in(0,1)$.
    Hence, for any $Q_m\in\mathcal{P}(\Theta^m)$,  
    \begin{IEEEeqnarray*}{rCl}
        d_\alpha(P_{\overline{Q}_{m}};\,P_0) &=& d_\alpha\left\{ P_{\int \widehat{Q}(\theta_{1:m}) \dt Q(\theta_{1:m})};\,P_0 \right\}\le \int d_\alpha(\overline{P}_{\theta_{1:m}};\,P_0)\dt Q_m(\theta_{1:m}).
    \end{IEEEeqnarray*}
    Applying $Q_m=\Qnm$ to the previous bound gives us
    \begin{equation*}
        \Esqr[0]{d_\alpha\{P_{\overline{Q}_{n,m}^{(\alpha)}};\,P_0\}} \lesssim \frac{\alpha}{m(1-\alpha)} + \frac{\alpha}{1-\alpha}\kld{P_0}{P_{Q^\star}} +\frac{m\,\kld{{Q^\star}}{\Pi}}{n(1-\alpha)},
    \end{equation*}
    whence the definition of $Q^\star$ provides us with the stated result.
\end{proof}
\subsection{Proof of~\Cref{prop:replica-gibbs}}
\begin{proof}
We first show that the marginal distribution of $\theta_{1:m}$ under $\widetilde Q_{n,m,r}$ is exactly $\Qnm$ in~\cref{eq:replica-log-target}.
Fix $\theta_{1:m}$ and consider a single replicated datum $\tilde x_{i,s}$, before conditioning on its value.
Since $z_{i,s}$ is uniform on $\{1,\ldots,m\}$ and $\tilde x_{i,s}\mid z_{i,s},\theta_{1:m}\sim P_{\theta_{z_{i,s}}}$, the law of total probability gives
\begin{equation*}
    p_{\tilde x_{i,s}\mid\theta_{1:m}}(x)
    =\sum_{j=1}^m\mathrm{pr}(z_{i,s}=j)\,p_{\theta_j}(x)
    =\frac{1}{m}\sum_{j=1}^m p_{\theta_j}(x),
\end{equation*}
that is, the equal-weight mixture density of the $m$ components.
By construction, $\tilde x_{i,1:r} = (\tilde x_{i,1},\ldots,\tilde x_{i,r})$ depend on $\theta_{1:m}$ only through their own labels $z_{i,1:r}$, which are mutually independent; 
hence $\tilde x_{i,1:r}$ are conditionally independent given $\theta_{1:m}$.
Since all $r$ of them are conditional on the same value $x_i$, their joint likelihood contribution is the $r$-fold product of the single-slot term above,
\begin{equation*}
    \prod_{s=1}^r p_{\tilde x_{i,s}\mid\theta_{1:m}}(x_i)
    =
    \left\{\frac{1}{m}\sum_{j=1}^m p_{\theta_j}(x_i)\right\}^{r}.
\end{equation*}
The blocks of replicates for different $i$ are independent of one another by the same argument, so multiplying over $i=1,\ldots,n$ gives the joint likelihood of the entire array of replicated data,
\begin{equation*}
    \prod_{i=1}^n
    \left\{
        \frac{1}{m}\sum_{j=1}^m p_{\theta_j}(x_i)
    \right\}^{r}.
\end{equation*}
Multiplying this likelihood by the prior density $\prod_{j=1}^m\pi(\theta_j)$ gives the unnormalised joint density of $\theta_{1:m}$ and the conditioning event,
\begin{equation*}
    \left\{\prod_{j=1}^m\pi(\theta_j)\right\}
    \prod_{i=1}^n
    \left\{
        \frac{1}{m}\sum_{j=1}^m p_{\theta_j}(x_i)
    \right\}^{r}.
\end{equation*}
Since $r=m\lambda_n/n$, this expression is term-for-term identical to the unnormalised density of $\Qnm$ in \cref{eq:replica-log-target}.
As both are densities with respect to the same dominating measure $\Pi^{\otimes m}$, their normalised versions coincide, so the marginal posterior of $\theta_{1:m}$ under the hierarchical model is exactly $\Qnm$.

Given $\theta_{1:m}$ and $\tilde x_{i,s}=x_i$, Bayes' rule gives
\begin{equation*}
    \mathrm{pr}(z_{i,s}=j\mid\theta_{1:m}, \tilde x_{i,s}=x_i)
    =\frac{p_{\theta_j}(x_i)}{\sum_{\ell=1}^m p_{\theta_\ell}(x_i)}
    =\rho_{i,j}.
\end{equation*}
For each fixed $i$, the variables $z_{i,1:r}$ are conditionally independent and have the same categorical distribution.
Their aggregate counts therefore satisfy \cref{eq:replica-count-conditional}.
The allocation variables are also conditionally independent across $i$, proving the conditional independence of $c_{1:n}$.

Finally, conditional on the allocations, grouping the likelihood factors assigned to component $j$ gives
\begin{equation*}
    \left\{\prod_{i=1}^n p_{\theta_j}(x_i)^{c_{i,j}}\right\}
    {\pi}(\theta_j).
\end{equation*}
This expression depends on no other component parameter, so the conditional posterior factorises over $j=1,\ldots,m$.
The same expression depends on the allocations only through their aggregate counts.
\end{proof}
\section{Additional results}
\label{app:additional-results}
\subsection{Exchangeability of the finite-particle Gibbs measure}
\begin{lemma}\label{lemma:exchangeable-marginals}
    {The minimiser, $\Qnm$, of \cref{eq:particle-objective} is exchangeable:
    for any $m$ and any $i,j,\,\Qnm^{(i)}=\Qnm^{(j)}$.}
\end{lemma}
\begin{proof}
    For a measurable map $f: \mathcal{X} \to \mathcal{Y}$ and a measure $\mu$ on  $\mathcal{X}$, the pushforward measure $f_\#\mu$ on $\mathcal{Y}$ is 
    defined by
    \begin{equation*}
        (f_\# \mu)(A) = \mu\{f^{-1}(A)\}
    \end{equation*}
    for all measurable $A \subseteq \mathcal{Y}$.
    Let $\mathfrak S_m$ denote the permutation group on $\{1,\dots,m\}$. 
    For $\sigma\in\mathfrak S_m$, define
    \begin{equation*}
    T_\sigma(\theta_1,\dots,\theta_m)
    =
    \{\theta_{\sigma(1)},\dots,\theta_{\sigma(m)}\},
    \qquad
    Q_m^{\mathrm{sym}}
    =
    \frac{1}{m!}\sum_{\sigma\in\mathfrak S_m} T_{\sigma\#}Q_m.
    \end{equation*}
    The map $\theta_{1:m}\mapsto \widehat Q_m(\theta_{1:m})$ is permutation invariant, hence so is $\Lnm(\theta_{1:m})$. 
    \begin{equation*}
    \int \Lnm(\theta_{1:m}) \dt Q_m^{\mathrm{sym}}(\theta_{1:m})
    =
    \int \Lnm(\theta_{1:m}) \dt Q_m(\theta_{1:m}).
    \end{equation*}
    Moreover, $\Pi^{\otimes m}$ is permutation invariant and the map $Q\mapsto \kld{Q}{\Pi^{\otimes m}}$ is convex in its first argument, so
    \begin{equation*}
    \kld{Q_m^{\mathrm{sym}}}{\Pi^{\otimes m}}
    \le
    \frac{1}{m!}\sum_{\sigma\in\mathfrak S_m}\kld{T_{\sigma\#}Q_m}{\Pi^{\otimes m}}
    =
    \kld{Q_m}{\Pi^{\otimes m}}.
    \end{equation*}
    Hence symmetrisation does not increase the objective in \cref{eq:particle-objective}. 
    Therefore, if $Q_m$ is a minimiser, then $Q_m^{\mathrm{sym}}$ is also a minimiser and is exchangeable. 
    Since $\Qnm$ is a minimiser of~\cref{eq:particle-objective} then its symmetrisation is also a minimiser.
    And since we know that the minimiser is unique, we know that $\Qnm$ is also its symmetrisation, and thus that it is exchangeable.
    This proves the claim.
\end{proof}
\subsection{Results for the logarithmic score}
\label{sec:log-score-beta-proof}
We first present the following technical condition.
\begin{assumption}\label{ass:log_score}
    There exists $Q^{\star}\in\arg\inf_{Q\in\mathcal{P}(\Theta)} d_{\mathrm{KL}}(P_0;\,P_Q),$ such that $\kld{Q^\star}{\Pi}<\infty$, and
    \begin{equation*}
        M(Q^\star)=\mathbb{E}_{0}\left\{ \mathbb{E}_{Q^\star}\left( \left[ \frac{p_{\theta}(X)}{\mathbb{E}_{Q^\star}\left\{ p_{\theta}(X)\right\}}\right]^4 \right)^{1/2}  \mathbb{E}_{Q^\star}\left( \log\left[ \frac{p_{\theta}(X)}{\mathbb{E}_{Q^\star}\left\{ p_{\theta}(X)\right\}}\right]^4 \right)^{1/4}+1\right\} < \infty.
    \end{equation*}
\end{assumption}
With this in place, we are ready to control the finite-particle bias in the case of the logarithmic score.
\begin{lemma}\label{lemma:log-score-beta}
    For the logarithmic scoring rule, under~\Cref{ass:log_score},
        $
        \mathbb{E}_0\{\beta_{n,m}(Q^\star)\}\lesssim 1/m.
        $
\end{lemma}
\begin{proof}
    First note that we can rewrite the bias as 
    \begin{flalign*}
        \beta_{n,m}(Q)&=-\frac{1}{n} \sum_{i=1}^n \int \left\{\log \frac{1}{m}\sum_{j=1}^{m}p_{\theta_j}(x_i) -\log p_Q(x_i)\right\}\dt Q^{\otimes m}(\theta_{1:m})\\
        &=-\frac{1}{n} \sum_{i=1}^n \int \left\{\log \frac{\frac{1}{m}\sum_{j=1}^{m}p_{\theta_j}(x_i)}{p_Q(x_i)}\right\}\dt Q^{\otimes m}(\theta_{1:m})\\
        &=\frac{1}{n} \sum_{i=1}^n \int\left\{\left(\frac{\frac{1}{m}\sum_{j=1}^{m}p_{\theta_j}(x_i)-p_Q(x_i)}{p_Q(x_i)}\right)-\log\left( \frac{\frac{1}{m}\sum_{j=1}^{m}p_{\theta_j}(x_i)}{p_Q(x_i)}\right)\right\}\dt Q^{\otimes m}(\theta_{1:m}),
    \end{flalign*}
    where the last line has added zero by noting that for all $x_i$ 
    \begin{equation*}
        \frac{1}{m}\int\sum_{j=1}^{m}p_{\theta_j}(x_i) \dt Q^{\otimes m}(\theta_{1:m})=\int p_\theta(x_i) \dt Q(\theta)=p_Q(x_i).
    \end{equation*}
    Define $u_i(\theta_{1:m})=\frac{1}{m}\sum_{j=1}^{m}p_{\theta_j}(x_i)/p_Q(x_i)$, and note that
    \begin{IEEEeqnarray*}{rCl}
        \beta_{n,m}(Q)&=&\frac{1}{n}\sum_{i=1}^n \int\left[\left\{\frac{\frac{1}{m}\sum_{j=1}^{m}p_{\theta_j}(x_i)-p_Q(x_i)}{p_Q(x_i)}\right\}-\log\left\{ \frac{\frac{1}{m}\sum_{j=1}^{m}p_{\theta_j}(x_i)}{p_Q(x_i)}\right\} \right] \dt Q^{\otimes m}(\theta_{1:m})\\
        &=&
        \frac{1}{n}\sum_{i=1}^n\int\left[\left\{u_i(\theta_{1:m})-1\right\}-\log\left\{u_i(\theta_{1:m})\right\}\right] \dt Q^{\otimes m}(\theta_{1:m})
        \\
        &=&
        -\frac{1}{n}\sum_{i=1}^n\int\left[\log\left\{u_i(\theta_{1:m})\right\} - \left\{u_i(\theta_{1:m})-1\right\}\right] \dt Q^{\otimes m}(\theta_{1:m})
        \\
        &\le& \frac{1}{n}\sum_{i=1}^n\int \left|\log\left\{ u_i(\theta_{1:m})\right\} - \{u_i(\theta_{1:m})-1\}\right| \dt Q^{\otimes m}(\theta_{1:m})
    \end{IEEEeqnarray*}
    Now, we apply Lemma 5 of \citet{cherief-abdellatif_asymptotics_2025} to the final equation, to obtain 
    \begin{flalign*}
        \beta_{n,m}(Q)&\le \frac{1}{n}\sum_{i=1}^n \int \left|\left\{u_i(\theta_{1:m})-1\right\}\cdot\log\left\{ u_i(\theta_{1:m})\right\}\right|\dt Q^{\otimes m}(\theta_{1:m})
    \end{flalign*}
    Applying Cauchy--Schwartz across each index $i$ we obtain
    \begin{multline}\label{eq:cauchy-schwartz}
        \beta_{n,m}(Q)\le \frac{1}{n}\sum_{i=1}^n\Bigg(\left[\int \left\{u_i(\theta_{1:m})-1\right\}^2 \dt Q^{\otimes m}(\theta_{1:m})\right]^{1/2} \\
        \times\left[\int \log\left\{u_i(\theta_{1:m})\right\}^2 \dt Q^{\otimes m}(\theta_{1:m})\right]^{1/2}\Bigg).
    \end{multline}
    We first consider the first term in the summand above.
    Recalling from the proof of~\Cref{lem:beta-nonnegative} that $\E[Q^{\otimes m}]{\widehat Q_m} = Q$, we have that $\Ecurl[Q^{\otimes m}]{u_i(\theta_{1:m})} = 1$ for each $i$, and thus 
    \begin{IEEEeqnarray*}{rCl}
        \Esqr[Q^{\otimes m}]{\{u_i(\theta_{1:m}) - 1\}^{2}} &=& \E[Q^{\otimes m}]{\left[u_i(\theta_{1:m}) - \Ecurl[Q^{\otimes m}]{u_i(\theta_{1:m})}\right]^{2}} \\
        &=& \mathrm{var}_{Q^{\otimes m}}\left\{u_i(\theta_{1:m})\right\} \\
        &=& \mathrm{var}_{Q^{\otimes m}}\left\{\frac{\frac{1}{m}\sum_{j=1}^{m}p_{\theta_j}(x_i)}{p_Q(x_i)}\right\} \\
        &=& \frac{1}{m^2}\sum_{j=1}^{m}\mathrm{var}_{Q}\left\{\frac{p_{\theta_j}(x_i)}{p_Q(x_i)}\right\} \\
        &=& \frac{1}{m}\mathrm{var}_{Q}\left\{\frac{p_{\theta}(x_i)}{p_Q(x_i)}\right\}.
    \end{IEEEeqnarray*}
    Plugging this back into~\cref{eq:cauchy-schwartz} yields
    \begin{equation*}
        \beta_{n,m}(Q) \leq \frac{1}{n}\sum_{i=1}^n\frac{1}{\sqrt{m}}\left[\mathrm{var}_{Q}\left\{\frac{p_{\theta}(x_i)}{p_Q(x_i)}\right\}\right]^{1/2} \left[\int \log\left\{u_i(\theta_{1:m})\right\}^2 \dt Q^{\otimes m}(\theta_{1:m})\right]^{1/2}.
    \end{equation*}
    Now, apply Lemma 6(ii) of \citet{cherief-abdellatif_asymptotics_2025} to the second term in the summand above, which gives, for some finite $K>0$, 
    \begin{multline*}
        \left[\int \log\left\{u_i(\theta_{1:m})\right\}^2 \dt Q^{\otimes m}(\theta_{1:m})\right]^{1/2} \\ \le \frac{K}{\sqrt{m}}\left[\mathbb{E}_{\theta\sim Q}\left\{\frac{p_\theta(x_i)}{p_Q(x_i)}\right\}^4 \right]^{1/4}\left\{1+\left(\mathbb{E}_{\theta\sim Q}\left[\left\{\log\frac{p_\theta(x_i)}{p_Q(x_i)}\right\}^4\right]\right)^{1/4}\right\}
    \end{multline*}
    Hence, we obtain, for $Q=Q^\star$, 
    \begin{IEEEeqnarray*}{rCl}
        \beta_{n,m}(Q^\star)&\leq&\frac{1}{n}\sum_{i=1}^n\Bigg(\frac{1}{\sqrt m}\left[\mathrm{var}_{Q^\star}\left\{\frac{p_{\theta}(x_i)}{p_{Q^\star}(x_i)}\right\}\right]^{1/2} \\
        &&\qquad \times \frac{K}{\sqrt{m}}\left[\mathbb{E}_{\theta\sim Q^\star}\left\{\frac{p_\theta(x_i)}{p_{Q^\star}(x_i)}\right\}^4 \right]^{1/4}\left\{1+\left(\mathbb{E}_{\theta\sim Q^\star}\left[\left\{\log\frac{p_\theta(x_i)}{p_{Q^\star}(x_i)}\right\}^4\right]\right)^{1/4}\right\}\Bigg)
     \end{IEEEeqnarray*}
    Take expectations with respect to $P_0$ on both sides and apply~\Cref{ass:log_score} to obtain the stated result.
\end{proof}
\subsection{Results for kernel scoring rules}\label{sec:kernel-scores}
General kernel scoring rules can be written as an expectation over a product measure in $\Theta$ for any $x\in\mathcal{X}$ \citep[Lemma 2 of][]{mclatchie_predictively_2025}, so that the ideal predictive score for any $Q\in\mathcal{P}(\Theta)$ is
\begin{equation}\label{eq:di-pred-score}
    S_n(P_Q) = \iint \left\{S_n(P_\theta) - \eta_n(\theta,\theta^\prime)\right\} \dt Q(\theta)\dt Q(\theta^\prime),
\end{equation}
where the function $\eta_n : \mathcal{X}\times\Theta^2\to\mathbb{R}_+$ satisfies $\eta_n(\theta,\theta^\prime)\ge 0$ and $\eta_n(\theta,\theta)=0$ for all $\theta,\theta^\prime\in\Theta$.

For $x,\,x^\prime\in\mathcal{X}$, let $\kappa(x,\,x^\prime)$ denote a positive-definite kernel with feature map $\phi:\,\mathcal{X}\to\mathcal{H}$, such that $\kappa(x,\,x^\prime) = \langle \phi(x),\,\phi(x^\prime) \rangle_{\mathcal{H}}$.
And for $\theta,\,\theta^\prime\in\Theta$, define $\mu(P_\theta)=\int \phi(x)\dt P_\theta(x)\in\mathcal{H}$ the kernel mean embedding of $P_\theta$, and the inner product $\langle\mu(P_\theta),\,\mu (P_{\theta^\prime})\rangle_{\mathcal{H}} = \iint \kappa(x,\,x^\prime)\dt P_{\theta}(x)\dt P_{\theta^\prime}(x^\prime)$.
Then the finite-particle  loss takes the form 
\begin{multline*}
    \Lnm^{(\kappa)}(\theta_{1:m}) \\ = \frac{1}{m^2} \sum_{j=1}^m\sum_{\ell=1}^m \left\{ \langle \mu(P_{\theta_j}), \mu(P_{\theta_\ell}) \rangle_{\mathcal{H}} - \langle \mu(P_{\theta_\ell}), \mu(\delta_{\{x_{1:n}\}}) \rangle_{\mathcal{H}} - \langle \mu(P_{\theta_j}), \mu(\delta_{\{x_{1:n}\}}) \rangle_{\mathcal{H}} \right\},
\end{multline*}
where $\delta_{\{x_{1:n}\}} = n^{-1}\sum_{i=1}\delta_{\{x_{i}\}}$.
In the case of (bounded) kernel scoring rules, we can exactly characterise the constant $M_n(Q)$ in~\Cref{ass:non-trivial-misspecification-exact}.
\begin{lemma}\label{lemma:di-beta}
    For general kernel scoring rules based on bounded kernels,
    \begin{equation*}
        M_n(Q) = \iint \eta_n(\theta,\theta^\prime) \dt Q(\theta) \dt Q(\theta^\prime);
    \end{equation*}
    and subsequently the particle bias of~\cref{eq:beta-defn} becomes
    \begin{equation*}
        \beta_{n,m}(Q) = \frac{1}{m}\iint \eta_n(\theta,\theta^\prime) \dt Q(\theta) \dt Q(\theta^\prime) \lesssim \frac{1}{m}.
    \end{equation*}
\end{lemma}
\begin{proof}
    By the form of the score in~\cref{eq:di-pred-score}
    \begin{IEEEeqnarray*}{rCl}
    \int \Lnm(\theta_{1:m}) \dt Q^{\otimes m}(\theta_{1:m})
    &=&
    \int S_n(P_\theta) \dt Q(\theta)
    -
    \frac{1}{m^2}\sum_{j,\ell=1}^m
    \int \eta_n(\theta_j,\theta_\ell) \dt Q^{\otimes m}(\theta_{1:m}) \\
    &=&
    \int S_n(P_\theta) \dt Q(\theta)
    -
    \frac{m(m-1)}{m^2}
    \iint \eta_n(\theta,\theta') \dt Q(\theta) \dt Q(\theta'),
    \end{IEEEeqnarray*}
    because the diagonal terms vanish. 
    On the other hand,
    \begin{equation*}
    S_n(P_Q)
    =
    \int S_n(P_\theta) \dt Q(\theta)
    -
    \iint \eta_n(\theta,\theta') \dt Q(\theta) \dt Q(\theta').
    \end{equation*}
    Subtracting the last two displays yields the formula.
    The final inequality follows since $\eta_n(\theta,\theta^\prime)<\infty$ for all $\theta,\theta^\prime\in\Theta$ under bounded kernels.
\end{proof}
\subsection{Extensions of~\Cref{lemma:kl-closeness}}\label{sec:kl-closeness-supp}
\begin{lemma}\label{prop:particle-free-energy-decomposition}
    Under the conditions of~\Cref{prop:entropy-curvature}, every \(Q_m\in\mathcal{P}(\Theta^m)\) such that \(Q_m\ll\Qn^{\otimes m}\) and the terms below are finite satisfies
    \begin{multline}\label{eq:particle-free-energy-decomposition}
        \left\{
        \int_{\Theta^m}\Lnm(\theta_{1:m})\dt Q_m(\theta_{1:m})
        +
        \frac{\kld{Q_m}{\Pi^{\otimes m}}}{m\lambda_n}
        \right\}
        -
        \left\{
        S_n(P_{\Qn})
        +
        \frac{\kld{\Qn}{\Pi}}{\lambda_n}
        \right\}
        \\
        =
        \frac{\kld{Q_m}{\Qn^{\otimes m}}}{m\lambda_n}
        +
        \int_{\Theta^m}
        \mathcal{B}_n\{\widehat Q_m(\theta_{1:m}),\Qn\}
        \dt Q_m(\theta_{1:m}).
    \end{multline}
    Consequently,
    \begin{equation}\label{eq:particle-free-energy-bound}
        \frac{\kld{\Qnm}{\Qn^{\otimes m}}}{m\lambda_n}
        +
        \int_{\Theta^m}
        \mathcal{B}_n\{\widehat Q_m(\theta_{1:m}),\Qn\}
        \dt\Qnm(\theta_{1:m})
        \leq
        \beta_{n,m}(\Qn).
    \end{equation}
\end{lemma}
\begin{proof}
    Write
    \begin{equation*}
        h_n(\theta)
        =
        \log\frac{\dt\Qn}{\dt\Pi}(\theta).
    \end{equation*}
    A change of reference measure gives
    \begin{multline*}
        \frac{\kld{Q_m}{\Pi^{\otimes m}}}{m\lambda_n}
        -
        \frac{\kld{\Qn}{\Pi}}{\lambda_n}
        \\
        =
        \frac{\kld{Q_m}{\Qn^{\otimes m}}}{m\lambda_n}
        +
        \frac{1}{\lambda_n}
        \int_{\Theta^m}
        \left\{
        \int_\Theta h_n(\theta)\dt\widehat Q_m(\theta_{1:m})(\theta)
        -
        \int_\Theta h_n(\theta)\dt\Qn(\theta)
        \right\}
        \dt Q_m(\theta_{1:m}).
    \end{multline*}
    By~\cref{eq:pro-fixed-point-log-density}, \(h_n=-\lambda_n g_n-c_n\).
    Since both \(\widehat Q_m(\theta_{1:m})\) and \(\Qn\) are probability measures, the constant cancels, and the preceding display becomes
    \begin{multline*}
        \frac{\kld{Q_m}{\Pi^{\otimes m}}}{m\lambda_n}
        -
        \frac{\kld{\Qn}{\Pi}}{\lambda_n}
        \\
        =
        \frac{\kld{Q_m}{\Qn^{\otimes m}}}{m\lambda_n}
        -
        \int_{\Theta^m}
        \int_\Theta g_n(\theta)
        \dt\{\widehat Q_m(\theta_{1:m})-\Qn\}(\theta)
        \dt Q_m(\theta_{1:m}).
    \end{multline*}
    By adding
    \begin{equation*}
        \int_{\Theta^m}
        \left\{S_n(P_{\widehat Q_m(\theta_{1:m})})-S_n(P_{\Qn})\right\}
        \dt Q_m(\theta_{1:m})
    \end{equation*}
    to both sides, we arrive at~\cref{eq:particle-free-energy-decomposition}.
    By the optimality of \(\Qnm\) in~\cref{eq:particle-objective}, we have that
    \begin{multline*}
        \int_{\Theta^m}\Lnm(\theta_{1:m})\dt\Qnm(\theta_{1:m})
        +
        \frac{\kld{\Qnm}{\Pi^{\otimes m}}}{m\lambda_n}
        \\
        \leq
        \int_{\Theta^m}\Lnm(\theta_{1:m})\dt\Qn^{\otimes m}(\theta_{1:m})
        +
        \frac{\kld{\Qn^{\otimes m}}{\Pi^{\otimes m}}}{m\lambda_n}.
    \end{multline*}
    Applying~\cref{eq:particle-free-energy-decomposition} to the product competitor and using the fact that
    \begin{equation*}
        \int_{\Theta^m}\widehat Q_m(\theta_{1:m})\dt\Qn^{\otimes m}(\theta_{1:m})
        =
        \Qn
    \end{equation*}
    gives
    \begin{IEEEeqnarray*}{rCl}
        &&
        \int_{\Theta^m}\Lnm(\theta_{1:m})\dt\Qn^{\otimes m}(\theta_{1:m})
        +
        \frac{\kld{\Qn^{\otimes m}}{\Pi^{\otimes m}}}{m\lambda_n}
        -
        S_n(P_{\Qn})
        -
        \frac{\kld{\Qn}{\Pi}}{\lambda_n}
        \\
        &=&
        \int_{\Theta^m}
        \mathcal{B}_n\{\widehat Q_m(\theta_{1:m}),\Qn\}
        \dt\Qn^{\otimes m}(\theta_{1:m})
        =
        \beta_{n,m}(\Qn),
    \end{IEEEeqnarray*}
    where the final equality follows because the linear term in~\cref{eq:pro-score-bregman} has expectation zero under \(\Qn^{\otimes m}\).
    Combining the last two displays with~\cref{eq:particle-free-energy-decomposition} evaluated at \(Q_m=\Qnm\) proves~\cref{eq:particle-free-energy-bound}.
    In particular,
    \begin{equation}\label{eq:particle-full-kl-bound}
        \kld{\Qnm}{\Qn^{\otimes m}}
        \leq
        m\lambda_n\beta_{n,m}(\Qn).
    \end{equation}
\end{proof}

\begin{corollary}
    Let \(\Qnm^{(1:k)}\) denote any \(k\)-particle marginal of the exchangeable measure \(\Qnm\).
    Then, under the conditions of~\Cref{prop:entropy-curvature}, whenever \(\beta_{n,m}(\Qn)<\infty\), for every \(1\leq k\leq m\),
    \begin{equation}\label{eq:particle-k-marginal-kl}
        \kldcurl{\Qnm^{(1:k)}}{\Qn^{\otimes k}}
        \leq
        \frac{m}{\lfloor m/k\rfloor}
        \lambda_n\beta_{n,m}(\Qn),
    \end{equation}
    and
    \begin{equation}\label{eq:particle-k-marginal-tv}
        \tvdcurl{\Qnm^{(1:k)}}{\Qn^{\otimes k}}
        \leq
        \sqrt{
        \frac{m\lambda_n\beta_{n,m}(\Qn)}
        {2\lfloor m/k\rfloor}
        }.
    \end{equation}
    In particular, if \(m\geq2k\), then
    \begin{equation}\label{eq:particle-k-marginal-simple}
        \kldcurl{\Qnm^{(1:k)}}{\Qn^{\otimes k}}
        \leq
        2k\lambda_n\beta_{n,m}(\Qn),
        \qquad
        \tvdcurl{\Qnm^{(1:k)}}{\Qn^{\otimes k}}
        \leq
        \sqrt{k\lambda_n\beta_{n,m}(\Qn)}.
    \end{equation}
    If in addition \(\beta_{n,m}(\Qn)\leq C_n/m\) for some finite \(C_n\), then
    \begin{equation}\label{eq:particle-k-marginal-simple-version-2}
        \kldcurl{\Qnm^{(1:k)}}{\Qn^{\otimes k}}
        \leq
        \frac{2k\lambda_nC_n}{m},
        \qquad
        \tvdcurl{\Qnm^{(1:k)}}{\Qn^{\otimes k}}
        \leq
        \sqrt{\frac{k\lambda_nC_n}{m}}.
    \end{equation}
\end{corollary}
\begin{proof}
    Set \(b_{m,k}=\lfloor m/k\rfloor\).
    By the data processing inequality,
    \begin{equation*}
        \kld{\Qnm}{\Qn^{\otimes m}}
        \geq
        \kldcurl{\Qnm^{(1:kb_{m,k})}}{\Qn^{\otimes kb_{m,k}}}.
    \end{equation*}
    Partition the first \(kb_{m,k}\) coordinates into \(b_{m,k}\) disjoint blocks of size \(k\), and let \(R_j\) denote the marginal law of block \(j\).
    The chain rule for relative entropy gives
    \begin{equation}
        \kldcurl{\Qnm^{(1:kb_{m,k})}}{\Qn^{\otimes kb_{m,k}}}
        =
        \kldcurl{\Qnm^{(1:kb_{m,k})}}{\otimes_{j=1}^{b_{m,k}}R_j}
        +
        \sum_{j=1}^{b_{m,k}}\kld{R_j}{\Qn^{\otimes k}}
        \geq
        \sum_{j=1}^{b_{m,k}}\kld{R_j}{\Qn^{\otimes k}}.
    \end{equation}
    Exchangeability gives \(R_j=\Qnm^{(1:k)}\) for every \(j\), and therefore
    \begin{equation*}
        \kld{\Qnm}{\Qn^{\otimes m}}
        \geq
        b_{m,k}\kldcurl{\Qnm^{(1:k)}}{\Qn^{\otimes k}}.
    \end{equation*}
    Combining this display with~\cref{eq:particle-full-kl-bound} proves~\cref{eq:particle-k-marginal-kl}, and Pinsker's inequality gives~\cref{eq:particle-k-marginal-tv}.
    If \(m\geq2k\), then \(m/\lfloor m/k\rfloor\leq2k\), which yields~\cref{eq:particle-k-marginal-simple}.
\end{proof}
\subsection{Posterior concentration}\label{sec:posterior-concentration}
In order to translate the previous results on the predictive distribution back to the posterior itself, we assume some further continuity conditions on the predictive which resemble those of \citet[Assumption 6]{mclatchie_predictively_2025}.
\begin{assumption}\label{ass:concentration}
    \begin{enumerate*}[
        label=(\roman*), 
        ref=\Cref{ass:concentration}(\roman*)
    ]
        \item The map $\theta\mapsto P_\theta$ is continuous and identifiable, and there exists $Q^\star\in\mathcal{P}(\Theta)$ such that $Q^\star\ll\Pi$ and $\int P_\theta \dt Q^\star(\theta) = P_0$;\label{ass:concentration-i}
        \item there exists a constant $C<\infty$ and $\varepsilon>0$ such that, for any $Q,Q^\prime\in\mathcal{P}(\Theta)$, if $\mathsf{d}_2(P_{Q},P_{Q^\prime})\le \varepsilon$, then  
        $
        \mathsf{d}_1(Q,Q^\prime)\le C\mathsf{d}_2(P_{Q},P_{Q^\prime})
        $; and \label{ass:concentration-ii}
        \item for some $a\ge1$ and a constant $K>0$, $\mathsf{d}_2(P_Q,P_{Q^\prime})^a\le K d_{\mathcal{S}}(P_Q,P_{Q^\prime})$ when applying~\Cref{thm:predictive-concentration}, or $\mathsf{d}_2(P_Q,P_{Q^\prime})^a\le K d_{\alpha}(P_Q,P_{Q^\prime})$ when applying~\Cref{thm:predictive-concentration-kl}.
        \label{ass:concentration-iii}
    \end{enumerate*}
\end{assumption}
\ref{ass:concentration-i} ensures that there exists an optimal posterior distribution whose predictive recovers the data-generating process exactly.
\ref{ass:concentration-ii} represents a continuity assumption on the map $Q\mapsto P_Q$.
\ref{ass:concentration-iii} states that the relevant predictive divergence is at least as strong as $\mathsf{d}_2$, at least up to some fixed distortion, and allows us to control $\mathsf{d}_2$ through the relevant predictive bound.
\begin{theorem}\label{thm:posterior-concentration} 
    If the assumptions of either~\Cref{thm:predictive-concentration} or \Cref{thm:predictive-concentration-kl} hold as well as~\Cref{ass:concentration}, then for any sequence $M_n\rightarrow\infty$ such that $M_n\gamma_n^{1/a}=o(1)$,   
    \begin{equation*}
        P_0\left\{\mathsf{d}_1(\barQnm,Q^\star) 
        \le 
        M_n\gamma_n^{1/a}\right\}\ge1-o(1)
    \end{equation*}
    where $\gamma_n\asymp m^{-1}+ \max\left\{{\log(n)n^{-1/2}}, r_n \right\}$ in the case of~\Cref{thm:predictive-concentration}, and $\gamma_n \asymp n^{-1/2}$ in the case of~\Cref{thm:predictive-concentration-kl}.
\end{theorem}
\begin{proof}
    Begin by noting that by~\ref{ass:concentration-i} there exists $Q^\star\in\mathcal{P}(\Theta)$ such that $P_{Q^\star} = P_0$, so that the first term in~\Cref{thm:predictive-concentration,thm:predictive-concentration-kl} is zero.
    Then, by~\ref{ass:concentration-iii} and Jensen's inequality, either
    \begin{IEEEeqnarray}{rCl}
        \Ecurl[0]{\mathsf{d}_2(P_{\barQnm},\,P_{Q^\star})} 
        &\leq&
        K^{1/a}\Ecurl[0]{d_{\mathcal{S}}(P_{\barQnm},\,P_{Q^\star})^{1/a}} \nonumber \\
        &\leq&
        K^{1/a}\left[\Ecurl[0]{d_{\mathcal{S}}(P_{\barQnm},\,P_{Q^\star})}\right]^{1/a}, \label{eq:posterior-jensen}
    \end{IEEEeqnarray}
    or
    \begin{IEEEeqnarray*}{rCl}
        \Esqr[0]{\mathsf{d}_2\{P_{\barQnm^{(\alpha)}},\,P_{Q^\star}\}}
        &\leq&
        K^{1/a}\Esqr[0]{d_{\alpha}\{P_{\barQnm^{(\alpha)}},\,P_{Q^\star}\}^{1/a}} \\
        &\leq&
        K^{1/a}\left(\Esqr[0]{d_{\alpha}\{P_{\barQnm^{(\alpha)}},\,P_{Q^\star}\}}\right)^{1/a},
    \end{IEEEeqnarray*}
    where the second inequality in each line follows from the concavity of the map $x\mapsto x^{1 / a}$ for $a \ge 1$.
    For simplicity, we shall henceforth just argue in the first case; all arguments apply verbatim to the second.
    For a sequence $M_n\rightarrow\infty$ such that $M_n\gamma_n^{1/a}=o(1)$, by Markov's inequality
    \begin{flalign*}
    P_0\left\{\mathsf{d}_2(P_{\barQnm},\,P_{Q^\star})>M_n\gamma_n^{1/a}\right\}
    &\le \frac{\Ecurl[0]{\mathsf{d}_2(P_{\barQnm},\,P_{Q^\star})}}{M_n\gamma_n^{1/a}}\\
    &\le \frac{K^{1/a}\left[\Ecurl[0]{d_{\mathcal{S}}(P_{\barQnm},\,P_{Q^\star})}\right]^{1 / a}}{M_n\gamma_n^{1/a}} \lesssim \frac{1}{M_n},
    \end{flalign*}
    where the second inequality follows from~\cref{eq:posterior-jensen}, and the final inequality from~\Cref{thm:predictive-concentration,thm:predictive-concentration-kl}. 
    Since $M_n\rightarrow\infty$, taking the complementary event yields 
    \begin{equation*}
        P_0\left\{\mathsf{d}_2(P_{\barQnm},\,P_{Q^\star})\le M_n\gamma_n^{1/a}\right\}\ge 1-o(1).
    \end{equation*}
    Since $M_n\gamma_n^{1/a}=o(1)$, for all $n$ sufficiently large we have $M_n\gamma_n^{1/a}\leq\varepsilon$, where $\varepsilon$ is the constant in~\ref{ass:concentration-ii}.
    Then, by~\ref{ass:concentration-ii},
    \begin{equation*}
        P_0\{\mathsf{d}_1({\barQnm},\,{Q^\star})\leq CM_n\gamma_n^{1/a}\}
        \geq 
        P_0\{\mathsf{d}_2(P_{\barQnm},\,P_{Q^\star})\leq M_n\gamma_n^{1/a}\}
        \geq 1-o(1).
    \end{equation*}
    Since $M_n$ is arbitrary up to a fixed multiplicative constant, the stated result follows.
\end{proof}
This now states that the common marginal of the finite-particle \ac{pro} posterior converges towards the predictively-optimal \textit{mixture} distribution  at the predictive rate, raised to the power $1 / a$. 
The exponent $a \geq 1$ measures how much weaker the discrepancy $\mathsf{d}_2$ on parameters is than the score divergence $d_{\mathcal{S}}$ that drives the predictive bound:
$a = 1$ transfers the predictive rate verbatim, while $a > 1$ slows it. 
As noted by \citet{mclatchie_predictively_2025}, $a = 2$ in the case of general kernel and log scoring rules, and this distortion is an artefact of the proof technique, not of the posterior itself.
\section{Exponential family parallel updates}\label{sec:exponential-family-updates}
For an exponential family model $P_\theta$ equipped with a conjugate prior from which we can draw samples, the particle conditionals in~\cref{eq:replica-component-conditional} are also closed-form densities. 
Parametrise the model density by its natural parameter $\eta$, sufficient statistic $T(x)$, carrier density $h(x)$, and log partition function $A(\eta)$, and suppose that
\begin{equation*}
    p_\eta(x)
    =
    h(x)\exp\left\{\eta\cdot T(x)-A(\eta)\right\},
    \quad
    \dt\Pi(\eta)
    \propto
    \exp\left\{\eta\cdot T_0-n_0A(\eta)\right\},
\end{equation*}
where $T_0$ and $n_0$ characterise the conjugate prior.
Writing $N_j=\sum_{i=1}^n c_{i,j}$, the conditional distribution of the $j$-th particle is then the same conjugate family with updated hyper-parameters
\begin{equation*}
    T_j
    =
    T_0+\sum_{i=1}^n c_{ij}T(x_i),
    \quad
    n_j
    =
    n_0+N_j.
\end{equation*}
Thus any such conjugate exponential family model yields exact particle updates that remain conditionally independent and parallelisable.

\section{Existence of predictively oriented posteriors and their finite-particle approximations}\label{sec:existence-uniqueness-app}
\subsection{Implicit Gibbs form of the predictively oriented posterior}
We begin by presenting in full the conditions necessary for the \ac{pro} posterior objective of~\cref{eq:pro-posterior-defn} to admit a unique minimiser.
\begin{lemma}\label{lemma:existence-uniqueness-full-pro}
    Suppose that Assumption~\ref{ass:convex} holds, that the score is continuous and bounded from below, and that the first variation of the predictive score with respect to the posterior exists.
    Assume further that the prior $\Pi$ admits a density with respect to the Lebesgue measure $\dt\Pi(\theta)\propto \exp\{-U(\theta)\}$ for a smooth function $U:\Theta\to\mathbb{R}$ such that:
    \begin{enumerate}[(i),nosep]
        \item there exist constants $c > 0$ and $c^\prime \in \mathbb{R}$ such that for all $\theta\in\Theta$,
        $
            \nabla_\theta U(\theta)\cdot\theta \geq c\lVert\theta\rVert^2 + c^\prime,
        $
        and,\label{ass:prior-coercivity}
        \item $\nabla_\theta U$ is Lipschitz continuous.\label{ass:prior-lipz}
    \end{enumerate}
    Then there exists a unique minimiser to~\cref{eq:pro-posterior-defn} characterised by the implicit equation
    \begin{equation}\label{eq:pro-self-gibbs}
        q_n
        (\theta)\propto 
        \exp\left\{\left. - \frac{\lambda_n}{n}\sum_{i=1}^n \frac{\delta S(P_Q, x_i)}{\delta Q}\right|_{Q = Q_n}(\theta) \right\} 
        \pi 
        (\theta),
    \end{equation}
    where $\{\delta S(P_Q, x)/\delta Q\}|_{Q = Q_n}$ denotes the functional derivative of $S(P_Q, x)$ with respect to $Q$ evaluated at $Q = {Q_n}$, which we assume exists.
\end{lemma}
\begin{proof}
    We have by Assumption~\ref{ass:convex} that the score $S(\cdot,\,x)$ is convex in its first argument.
    In combination with the fact that the mapping $Q\mapsto P_Q$ is linear, we have the mapping $Q\mapsto S(P_Q,\,x)$ is convex.
    The existence and uniqueness of a minimiser is then provided under the same conditions in Proposition 2.5 of \citet{hu_mean-field_2021}, and we point the reader to their Section 4 for a proof. 
    
    To characterise the form of this minimiser, we take the functional derivative of the objective of~\cref{eq:pro-posterior-defn} with respect to $Q$, set it equal to zero, and solve for $Q$.
    We begin by noting that the functional derivative of the Kullback--Leibler regulariser evaluated at $\theta$ is
    \begin{equation}\label{eq:func-deriv-regulariser}
        1 + \log\left\{\frac{\dt Q}{\dt\Pi}(\theta)\right\},
    \end{equation}
    while that of the score term is
    \begin{equation}\label{eq:func-deriv-score-wrt-Q} 
        \frac{\lambda_n}{n} \sum_{i=1}^n \frac{\delta S(P_Q, x_i)}{\delta Q}(\theta).
    \end{equation}
    By the linearity of functional derivatives, setting the sum of~\cref{eq:func-deriv-regulariser,eq:func-deriv-score-wrt-Q} equal to zero and removing constant terms, we find that $Q_n$ must satisfy
    \begin{equation*}
        0 = \frac{\lambda_n}{n} \sum_{i=1}^n \left.\frac{\delta S(P_Q, x_i)}{\delta Q}\right|_{Q = Q_n}(\theta) + \log\left\{\frac{\dt Q_n}{\dt\Pi}(\theta)\right\}
    \end{equation*}
    whence rearranging yields the stated result.
\end{proof}
The conditions of~\ref{ass:prior-coercivity} and~\ref{ass:prior-lipz} are easily satisfied, for instance, by choosing the prior to be a multivariate standard Gaussian \citep{hu_mean-field_2021}.
And the condition that the score be bounded from below is immediate, for example, with kernel scoring rules under bounded kernels \citep[see also][Theorem 1]{shen_prediction-centric_2025}.
In the case of the log score, we refer the reader to \citet[Proposition 1]{liu_detecting_2025} who establish the necessary conditions.
\subsection{A demonstrative example}
Below we provide an example which shows that the finite-particle \ac{pro} posterior can exist, even when the limiting \ac{pro} posterior does not.
Let $\Theta=\mathbb R$, $\Pi=\mathcal{N}(0,1^2)$, and suppose there is a single observation $x_1=0$.
We then consider the model $P_\theta = \normal(0,e^{-2\theta})$
with, in particular, 
\begin{equation*}
    p_\theta(0)
    =
    \frac{1}{\sqrt{2\pi e^{-2\theta}}}
    =
    \frac{e^\theta}{\sqrt{2\pi}}.
\end{equation*}
Then the prior predictive density at the observed point is finite and positive:
\begin{equation*}
    \int p_\theta(0)\dt\Pi(\theta)
    =
    \frac{1}{\sqrt{2\pi}} \int e^\theta \dt \Pi(\theta)%
    =
    \frac{e^{1/2}}{\sqrt{2\pi}}.
\end{equation*}
For a learning rate $\lambda>0$, the ideal log-score objective is
\begin{equation}\label{eq:eu-counterexample-mean-field-app}
    -\log\left\{\int p_\theta(0)\dt Q(\theta)\right\}
    +
    \frac{\kld{Q}{\Pi}}{\lambda}.
\end{equation}
We presently show that this objective is unbounded from below, leading to an ill-specified programme with no minimiser.

We begin by showing that 
the log predictive term
is unbounded below:
for $R>0$, let $A_R=[R,R+1]$ and $\varepsilon_R=e^{-R/2}$, and define the measure
\begin{equation*}
    Q_R
    =
    (1-\varepsilon_R)\Pi
    +
    \varepsilon_R\Pi(\cdot\mid A_R).
\end{equation*}
Since for all $\theta\in A_R$ we have that $\theta\geq R$, it follows that
\begin{IEEEeqnarray*}{rCl}
    \int_{\mathbb R}p_\theta(0)\dt Q_R(\theta)
    &\geq&
    \varepsilon_R
    \int_{A_R}p_\theta(0)\dt\Pi(\theta\mid A_R) \\
    &=&
    \frac{\varepsilon_R}{\sqrt{2\pi}}
    \int_{A_R}e^\theta\dt\Pi(\theta\mid A_R) \\
    &\geq&
    \frac{\varepsilon_R e^R}{\sqrt{2\pi}} \\
    &=&
    \frac{e^{R/2}}{\sqrt{2\pi}}.
\end{IEEEeqnarray*}
Taking the negative of the logarithm of both sides, 
\begin{equation}\label{eq:eu-counterexample-score-bound-app}
    -\log\left\{\int_{\mathbb R}p_\theta(0)\dt Q_R(\theta)\right\}
    \leq
    -\frac{R}{2}+\frac12\log(2\pi).
\end{equation}
As $R\to\infty$ the left-hand side tends to $-\infty$ linearly.

We now bound the entropic regularisation term.
The conditional law $\Pi(\cdot\mid A_R)$ is absolutely continuous with respect to $\Pi$. Moreover, 
\begin{equation*}
    \kldcurl{\Pi(\cdot\mid A_R)}{\Pi}
    =
    -\log\{\Pi(A_R)\}
\end{equation*}
so that by the joint convexity of the Kullback--Leibler divergence,
\begin{equation}\label{eq:eu-counterexample-entropy-mix-app}
    \kld{Q_R}{\Pi}
    \leq
    -\varepsilon_R
    \log\{\Pi(A_R)\}.
\end{equation}
For the standard Gaussian prior and $R\geq0$,
\begin{equation*}
    \Pi(A_R)
    =
    \frac{1}{\sqrt{2\pi}}\int_R^{R+1}e^{-\theta^2/2}\dt \theta
    \geq
    \frac{1}{\sqrt{2\pi}}e^{-(R+1)^2/2},
\end{equation*}
because $e^{-\theta^2/2}\geq e^{-(R+1)^2/2}$ for all $\theta\in[R,R+1]$.
It follows then that
\begin{equation*}
    -\log\{\Pi(A_R)\}
    \leq
    \frac{(R+1)^2}{2}+\frac12\log(2\pi)
    =
    O(R^2),
\end{equation*}
so that in turn, due to~\cref{eq:eu-counterexample-entropy-mix-app},
\begin{equation}\label{eq:eu-counterexample-entropy-vanishes-app}
    \kld{Q_R}{\Pi}
    \leq
    e^{-R/2}O(R^2)
    \to0.
\end{equation}
Combining~\cref{eq:eu-counterexample-score-bound-app,eq:eu-counterexample-entropy-vanishes-app},
\begin{equation*}
    -\log\left\{\int p_\theta(0)\dt Q_R(\theta)\right\}
    +
    \frac{\kld{Q_R}{\Pi}}{\lambda}
    \leq
    -\frac{R}{2}+O(1)+\frac{e^{-R/2}O(R^2)}{\lambda}
    \to
    -\infty,
\end{equation*}
in the limit $R\to\infty$.
Thus, the ideal \ac{pro} objective of~\cref{eq:pro-posterior-defn} is unbounded below, and has no minimiser.

We now turn our attention to the finite-particle objective for the same model.
Fix $m\geq1$ and define
\begin{equation*}
    \mathsf{L}_m(\theta_{1:m})
    =
    -\log\left\{
        \frac{1}{m}\sum_{j=1}^{m}p_{\theta_j}(0)
    \right\}
    =
    \frac12\log(2\pi)
    -
    \log\left(\frac{1}{m}\sum_{j=1}^{m}e^{\theta_j}\right).
\end{equation*}
The finite-particle variational objective is
\begin{equation*}
    \inf_{Q_m\in\mathcal P(\mathbb R^m)}
    \left\{
        \int\mathsf{L}_m(\theta_{1:m})\dt Q_m(\theta_{1:m})
        +
        \frac{\kld{Q_m}{\Pi^{\otimes m}}}{m\lambda}
    \right\},
\end{equation*}
and the normalising constant of its minimiser is
\begin{IEEEeqnarray*}{rCl}
    Z_{m,\lambda}
    &=&
    \int
    \exp\{-m\lambda\mathsf{L}_m(\theta_{1:m})\}
    \dt\Pi^{\otimes m}(\theta_{1:m}) \\
    &=&(2\pi)^{-(m\lambda)/2} \int \left(\frac{1}{m}\sum_{j=1}^{m}e^{\theta_j}\right)^{m\lambda} \dt \Pi^{\otimes m}(\theta_{1:m}),
\end{IEEEeqnarray*}
where the last line is finite for every fixed $m$ and $\lambda$.
Indeed, for all $a_1,\ldots,a_m\geq0$ and every $m\lambda>0$, there is a finite constant $C_{m,\lambda}$ such that
\begin{equation*}
    \left(\frac{1}{m}\sum_{j=1}^{m}a_j\right)^{m\lambda}
    \leq
    C_{m,\lambda}\sum_{j=1}^{m}a_j^{m\lambda}.
\end{equation*}
Applying this inequality with $a_j=e^{\theta_j}$ gives
\begin{IEEEeqnarray*}{rCl}
    \int \left(\frac{1}{m}\sum_{j=1}^{m}e^{\theta_j}\right)^{m\lambda} \dt \Pi^{\otimes m}(\theta_{1:m}) 
    &\leq& 
    C_{m,\lambda}\sum_{j=1}^{m}\int e^{m\lambda\theta_j} \dt \Pi^{\otimes m}(\theta_{1:m}) \\
    &=&
    C_{m,\lambda}m e^{(m\lambda)^2/2}
\end{IEEEeqnarray*}
which is in turn finite.
Since the integrand is strictly positive and finite, so too is the normalising constant, $0<Z_{m,\lambda}<\infty$, and thus the particle \ac{pro} posterior exists.
\section{Comparison with propagation of chaos assumptions}\label{sec:poc-supp}
\subsection{A demonstrative result}
Suppose that 
the prior admits a density 
\(
\dt\Pi(\theta)
    \propto
    \exp\{-U(\theta)\}
\)
with respect to Lebesgue measure.
Then the assumptions required for the propagation of chaos results are as follows.
\begin{assumption} \label{ass:poc}
    Let $\mathcal{P}_{2}(\Theta)$ denote the probability measures on $\Theta$ with finite second moment and a Lebesgue density.
    The following conditions hold:
    \begin{enumerate}[
        nosep,
        label=(\roman*), 
        ref=\Cref{ass:poc}(\roman*)
    ]
        \label{ass:poc-convex}
        \item the map $Q\mapsto S_n(P_Q)$ is Wasserstein smooth:
        there exists $L_n<\infty$ such that, for all $\theta,\theta^\prime\in\Theta$ and all $Q,Q^\prime\in\mathcal{P}_{2}(\Theta)$,
        \begin{equation*}
            \left\|
            \nabla_{\mathrm{W}_2}S_n(P_Q)(\theta)
            -
            \nabla_{\mathrm{W}_2}S_n(P_{Q^\prime})(\theta^\prime)
            \right\|
            \leq
            L_n
            \left\{
                \|\theta-\theta'\|
                +
                d_{\mathrm{W}_1}(Q,Q^\prime)
            \right\}.
        \end{equation*}
        Here $\nabla_{\mathrm{W}_2}$ denotes the Wasserstein-2 gradient, $d_{\mathrm{W}_1}$ is the Wasserstein-1 distance, and $\|\cdot\|$ is the Euclidean norm. 
        \label{ass:poc-smooth}
        \item The prior-tilted proximal Gibbs measures 
        satisfy a log-Sobolev inequality with a constant uniform in the proximal centre.
        More precisely, for each $\mu\in\mathcal{P}_{2}(\Theta)$ define
        \begin{equation*}%
            { q_n^\mu }
            (\theta)
            \propto
            \exp\left\{
                -\lambda_n
                \frac{\delta S_n(P_Q)}{\delta Q}\bigg|_{Q = \mu}(\theta)
            \right\}
            { \pi }
            (\theta).
        \end{equation*}
        There exists $\bar C_{\mathrm{LSI},n}<\infty$ such that
        \(
            \max\left\{C_{\mathrm{LSI}}(\Qn)
            ,
            \sup_{\mu\in\mathcal{P}_{2}(\Theta)}
            C_{\mathrm{LSI}}(Q_n^\mu)
            \right\}
            \leq
            \bar C_{\mathrm{LSI},n}.
        \)
        \label{ass:poc-lsi}
    \end{enumerate}
\end{assumption}
Under these conditions one obtains the following posterior-level comparison for the finite-particle target, similar to our~\Cref{lemma:kl-closeness}, by the propagation of chaos route.
\begin{proposition} \label{prop:poc-reference-route}
    Suppose that~\Cref{ass:convex,ass:poc} holds. Then, for all $m$ sufficiently large,
    \begin{equation*}
        \kld{\bar Q_{n,m}}{\Qn}
        \lesssim
        \frac{
            L_n\bar C_{\mathrm{LSI},n}d\lambda_n
        }{
            m
        }.
    \end{equation*}
\end{proposition}
This is essentially Theorem~24 of \citet{kook2024sampling}, specialised to identical one-particle potentials $V_n=U/\lambda_n$ and $\sigma_n^2=2/\lambda_n$.
It is instructive to compare~\Cref{ass:poc} with the assumptions used in our prediction-centric analysis.
Verifying~\ref{ass:poc-smooth} reduces to showing that the first variation of the score has a gradient with respect to the parameter that is jointly Lipschitz in both $\theta$ and $Q$.
This condition is relatively mild, but does not follow immediately from~\Cref{ass:convex} (see the following subsection for discussion on the applicability to the scores considered above).
The log-Sobolev condition of~\ref{ass:poc-lsi}, however, does not generally follow from our assumptions.
One sufficient route is a bounded-oscillation perturbation argument:
if the prior $\Pi$ satisfies a log-Sobolev inequality and $S_n(P_Q)$ satisfies a bounded-oscillation perturbation inequality
then the Holley--Stroock perturbation principle allows us to transport the log-Sobolev constant of $\Pi$ to an upper bound on the log-Sobolev constant of $Q_n^\mu$ \citep{holley_logarithmic_1987}.
This condition holds for bounded kernels on compact parameter spaces, but may yield a worsening log-Sobolev constant when the learning rate $\lambda_n$ grows with $n$.
Moreover, for the logarithmic score, this condition typically cannot hold without additional restrictions, particularly on noncompact parameter spaces.
\subsection{On the assumptions of~\Cref{prop:poc-reference-route}}
Consider, for example, the squared maximum mean discrepancy. 
Then, up to the normalisation convention used for $S_n$, the first variation is
\begin{equation*}
    \frac{\delta S_n(P_Q)}{\delta Q}(\theta)
    =
    2
    \left\langle
        \mu(P_\theta),
        \mu(P_Q)-\mu(\delta_{\{x_{1:n}\}})
    \right\rangle_{\mathcal H}.
\end{equation*}
If $\theta\mapsto \mu(P_\theta)$ has uniformly bounded and Lipschitz first derivative as a map into $\mathcal H$, then this function has the required smoothness.
This is not an unreasonable condition for smooth simulators and smooth bounded kernels, but it excludes the non-differentiable forward maps which are permitted by our approach.

For the logarithmic score, the same condition is rather more restrictive.
In this case, the first variation is
\begin{equation*}
    \frac{\delta S_n(P_Q)}{\delta Q}(\theta)
    =
    -\frac{1}{n}\sum_{i=1}^{n}
    \frac{p_\theta(x_i)}{p_Q(x_i)}.
\end{equation*}
Thus, uniform Lipschitzness requires, at minimum, smooth likelihoods and lower-envelope control of the mixture densities $p_Q(x_i)$ over the relevant class of mixing distributions.
Such assumptions can be imposed in restricted models, but they were not required for the finite-particle bias argument in~\Cref{lemma:log-score-beta}.
They also fail immediately in examples where the forward map is non-differentiable in $\theta$, such as the nuclear-physics application in~\Cref{sec:physics}.
\section{Further experimental details}
\subsection{Gaussian location model}\label{sec:numerical-demonstration-supp}
The underlying data generating measure is the bimodal Gaussian $P_0 = 0.5\mathcal N(-1.5, \sigma^2) + 0.5\mathcal N(1.5, \sigma^2)$, from which we draw $n = 100$ observations $x_{1:n}$;
the likelihood is $P_\theta = \mathcal N(\theta, \sigma^2)$ with $\sigma^2 = 0.2^2$;
and the prior is $\Pi = \mathcal N(0, \tau^2)$ with $\tau^2 = 5^2$.
We use a learning rate $\lambda_n = n$ and $m = 2$ particles.

We evaluate $\dt \barQnm$ on a uniform grid $\mathcal G = \{g_1, \dots, g_G\}$ of $G = 10{,}000$ points
on $[-6, 6]$, approximating the integral over the remaining particles $g_{2:m}$ by a sum:
\begin{equation*}
  \overline q_{n,m} 
  (\theta)
  \propto
  \sum_{g_{2:m} \in \mathcal G^{m-1}}\left(
  \left\{ \prod_{j=2}^{m} 
  { \pi }
  (g_j) \right\}
  \prod_{i=1}^{n}
  \left[
    \tfrac1m \left\{\, p_\theta(y_i) + \textstyle\sum_{j=2}^{m} p_{g_j}(y_i)\,\right\}
  \right]^{2}
  \right)
   \pi 
  (\theta),
\end{equation*}
normalised so that $\sum_{\theta\in\mathcal G} \overline q_{n,m} (\theta) = 1$.
Because the summand is invariant to permutations of $g_{2:m}$, we sum over multisets of $\{g_2,\dots,g_m\}$ rather than ordered tuples, weighting each by its multinomial multiplicity;
this reduces the cost from $G^{m-1}$ to $\binom{G + m - 2}{m-1}$ terms while leaving the value unchanged.

Initial particles for the \ac{wgf} and Metropolis adjusted Langevin were drawn from the prior, and use $4{,}000$ warm-up and $4{,}000$ retained iterations.
Both used a step size of $\dt t = \sigma^2/n = 4\times 10^{-4}$, and $67\%$ of proposals were accepted by the Metropolis step.
For the blocked Gibbs sampler, $\alpha = \lambda_n m / n = 2$, and again $4{,}000$ draws were used for warm-up with a further $4{,}000$ retained.
\subsection{Sensitivity to step size}\label{sec:step-size}
\begin{figure}
    \centering
    \includegraphics[width=\linewidth]{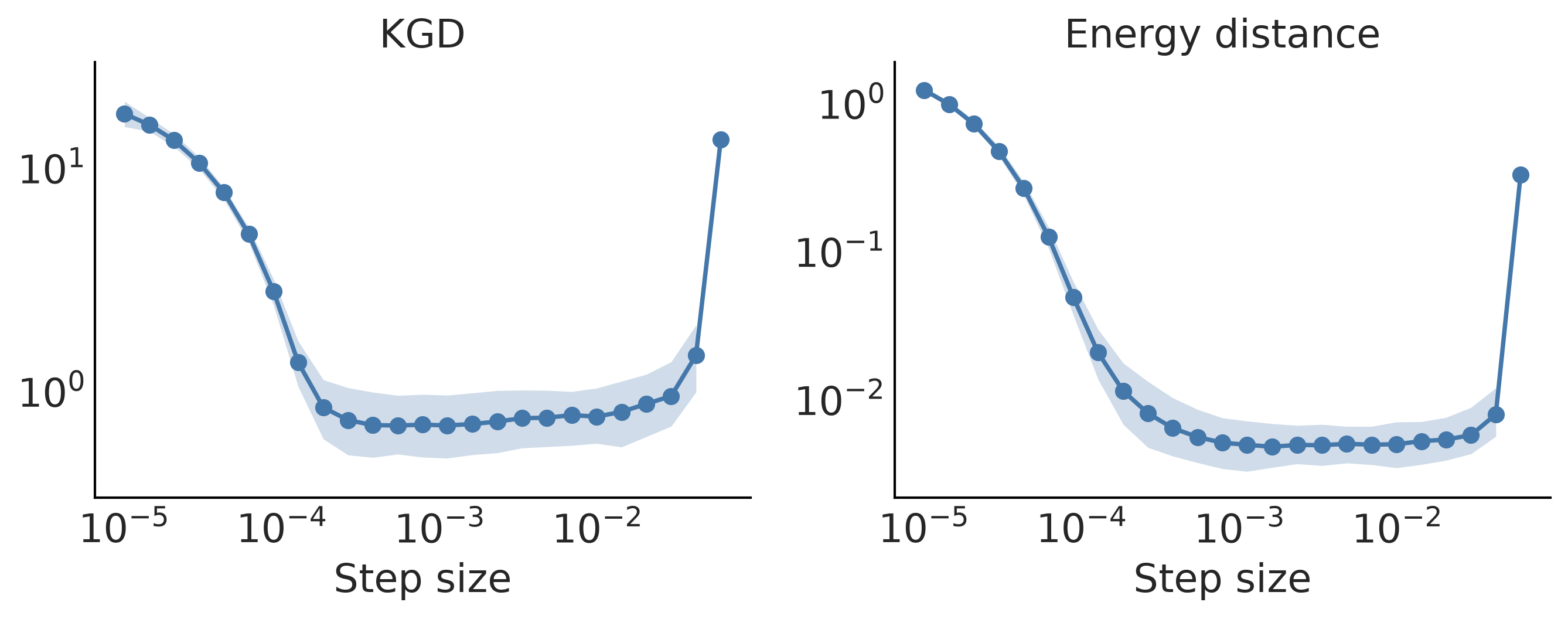}
    \caption{\textit{Simulated bivariate normal location example}\/.
    The left facet shows the kernel gradient discrepancy (KGD) of the final iteration of a \ac{wgf} under different step sizes; 
    the right facet shows the energy distance of those same particles to samples from a gold-standard run of Metropolis-within-Gibbs.
    Solid lines represent the median values over the $200$ replicates, and the ribbons the $10$--$90\%$ quantiles.}
    \label{fig:step-size}
\end{figure}
In the following experiment we demonstrate the sensitivity of unadjusted Monte Carlo schemes to the choice of step size.
We sample $n = 50$ data uniformly from a disc of radius $4$ in $\mathbb{R}^2$. 
To these we fit a two-dimensional normal location model $p_\theta(y_i) = \normal(y_i;\,\theta, I_2)$ with fixed variance, and prior $\theta \sim \normal(0, 5^2 I_2)$. 
We consider the particle \ac{pro} posterior under the logarithmic score with $\lambda_n = 1$ and $m = 100$ particles. 
We sample $1{,}000$ iterations from a \ac{wgf} under $25$ difference step sizes in $[10^{-5}, 6 \times 10^{-2}]$. 
Each run is initialised independently from $\normal(0, 0.1^2 I_2)$ and repeated $200$ times.
Following \citet{chazal_computable_2025}, we report the median and $10$--$90\%$ quantiles of the kernel gradient discrepancy (KGD) of the final particle configuration of the flow, and the energy distance of the particle marginal to that of gold-standard run from Metropolis-within-Gibbs in~\Cref{fig:step-size}.
For small step sizes the particles remain close to their initialisations, while a large step size incurs discretisation bias and ultimately instability, with almost half of the runs diverging at step size $6 \times 10^{-2}$.
\subsection{Compartmental models}\label{sec:sir-supp}
The total population is assumed fixed at $N = 1{,}000$ with initial state $(S_0, I_0, R_0) = (999, 1, 0)$.
The underlying differential equation is integrated from $t_0 = 0$ to $t_1 = 60$ with solver step $\mathrm{d}t = 0.1$.
Only the infected compartment is observed under additional Gaussian noise with standard deviation $\sigma = 3.0$.
The prior is placed on the log of the reproduction number and the log of the recovery rate: 
$\log R_0 = \log\beta - \log\gamma \sim \mathcal{N}(0.92, 0.1^2)$ and $\log\gamma \sim \mathcal{N}(-2.3, 0.2^2)$, so that a prior on $\beta$ is implied by the relationship $\beta = R_0\gamma$.
The Gibbs and \ac{pro} posteriors were fit with the maximum mean discrepancy estimated with $64$ Monte Carlo samples, where the kernel lengthscale was set to the plug-in estimate $\hat\sigma^2$ discussed in~\Cref{sec:sir} and $\lambda_n = 8$.
The \ac{pro} posterior was fit with $m=16$ particles. 
All posteriors are sampled with Hamiltonian Monte Carlo: 
the Bayes and Gibbs posteriors used $4$ chains with $4{,}000$ warm-up and $1{,}000$ retained draws each, and the \ac{pro} posterior used a single chain with $1{,}000$ warm-up and $1{,}000$ retained draws.
\Cref{fig:sir-lockdown-r0} shows the Bayes, Gibbs, and \ac{pro} posterior distributions over the reproduction number.
\begin{figure}
    \centering
    \includegraphics[width=\linewidth]{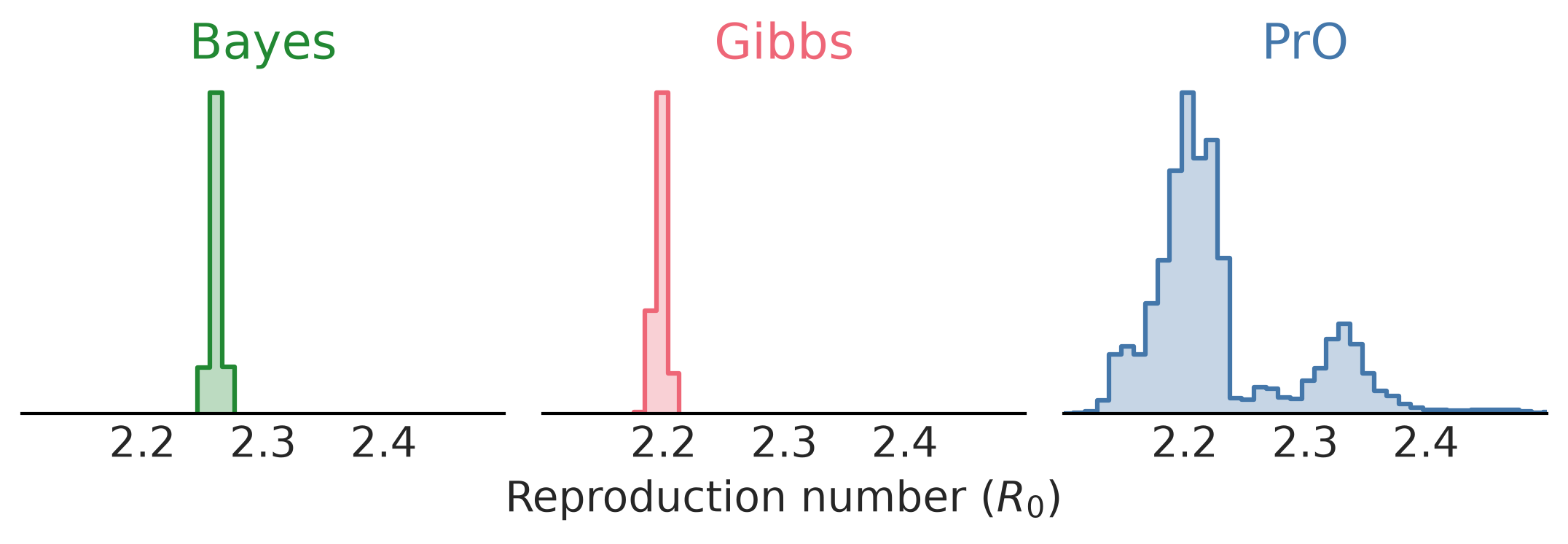}
    \caption{\textit{Compartmental model with unmodelled infection rate shift}\/ (\ref{dgp:lockdown}). 
    \textcolor{bayescolour}{\bfseries\sffamily Bayes}, \textcolor{gibbscolour}{\bfseries\sffamily Gibbs}, and \textcolor{procolour}{\bfseries\sffamily \ac{pro}} posterior reproduction number $R_0 = \beta / \gamma$: 
    the true data-generating process is a mixture of $R_0^{\rm{pre}} = 2.5$ pre-lockdown and $R_0^{\rm{post}} = 2.0$ post-lockdown.}
    \label{fig:sir-lockdown-r0}
\end{figure}
\subsection{Spatial crime modelling}\label{sec:guerry-supp}
The spatial-lag Poisson model is fit with independent standard Gaussian priors on all coefficients, and learning rate $\lambda_n = n$.
Both posteriors are evaluated by leave-one-d\'epartement-out cross-validation over the $85$ folds, refitting the model on each training fold. 
The Bayes posterior is sampled with Hamiltonian Monte Carlo using four chains of $1{,}000$ warm-up and $1{,}000$ retained draws each; 
the PrO posterior uses a single chain of $1{,}000$ warm-up and $4{,}000$ retained draws.

The middle panel of~\Cref{fig:guerry} reports the difference between the leave-one-region-out cross-validation score of the \ac{pro} posterior and the Bayes posterior.
That is, for $\overline{Q}_{-i,m}$ the \ac{pro} posterior fit with $m$ particles and the $i$-th region left out, and $\Pi_{-i}$ the leave-one-out Bayes posterior, we compute
\begin{equation*}
    \Delta\mathrm{elpd}_i(m) = \left\{\log p_{\overline{Q}_{-i,m}}(y_i\mid y_{-i}) - \log p_{\Pi_{-i}}(y_i\mid y_{-i})\right\}.
\end{equation*}
We then report the mean $\sum_{i=1}^n\Delta\mathrm{elpd}_i(m)$ and standard error $\sqrt{n\,\mathrm{var}\{\Delta\mathrm{elpd}_i(m)\}}$.
The dashed green line at zero denotes equivalent performance to the Bayes posterior; 
values above it indicate better held-out predictive performance than Bayes.
\subsection{Nuclear physics}\label{sec:physics-supp}
Parameters that are constrained to be strictly positive were transformed onto the log scale.
The (possibly transformed) optical-model parameters are given Gaussian priors centred at the point-estimates derived by \citet{koning_local_2003}, and scaled by a factor of $0.4$ of that mean value.
Gibbs and \ac{pro} posteriors are fit with a learning rate of $\lambda_n = 14$. 
All posteriors are sampled with random walk Metropolis--Hastings where the step size is adapted with Robbins--Monro stochastic approximation \citep{andrieu_tutorial_2008} to target an acceptance rate of 0.234. 
The Bayes posterior uses $4$ chains with a total of $32{,}000$ samples, of which the last $16{,}000$ are retained; 
the Gibbs posterior uses $4$ chains with a total of $96{,}000$ samples of which the last $64{,}000$ are retained.
The \ac{pro} posterior uses a single chain with $32{,}000$ total samples of which the last $16{,}000$ are retained.
Both the Gibbs and \ac{pro} posteriors are fit using the maximum mean discrepancy which is estimated from $128$ Monte Carlo draws.
The observations exist across several magnitudes, and thus a single kernel with a single lengthscale can fail to deliver reliable inference.
Instead, we fit those posteriors using an equally-weighted mixture-of-kernels \citep[see, e.g., Section 4.1 of][for a review]{briol_statistical_2019} approach so that $\kappa(x,y) = \sum_{i=1}^n \kappa(x,y;\,s_i^2)$ where the component kernel lengthscales are the squared statistical errors at each angle $s_i^2$.

\end{document}